\documentclass[acmsmall]{acmart}

\AtBeginDocument{\providecommand\BibTeX{{\normalfont B\kern-0.5em{\scshape i\kern-0.25em b}\kern-0.8em\TeX}}}

\setcopyright{acmcopyright}
\copyrightyear{2019}
\acmYear{2019}

\acmJournal{TECS}
\acmVolume{1}
\acmNumber{1}
\acmArticle{1}
\acmMonth{1}
\usepackage{pifont}
\usepackage{xcolor}
\usepackage{subcaption}
\usepackage{booktabs}
\usepackage[T1]{fontenc}
\usepackage{algorithm}
\usepackage{algorithmicx}
\usepackage{algpseudocode}
\usepackage{cleveref}
\usepackage{todonotes}
\usepackage{listings}
\usepackage{hyperref}
\hypersetup{colorlinks=true,           allcolors=blue!70!black,   pdfstartview=Fit,          breaklinks=true,
}

\newtheorem{remark}{Remark}
\newtheorem{assumption}{Assumption}
\newtheorem{example}{Example}

\crefname{section}{Section}{Sections}
\crefname{subsection}{Section}{Sections}
\crefname{definition}{Definition}{Definitions}
\crefname{proposition}{Proposition}{Propositions}
\crefname{lemma}{Lemma}{Lemmas}
\crefname{theorem}{Theorem}{Theorems}
\crefname{corollary}{Corollary}{Corollaries}
\crefname{example}{Example}{Examples}
\crefname{figure}{Figure}{Figures}
\crefname{assumption}{Assumption}{Assumptions}
\crefname{remark}{Remark}{Remarks}
\crefname{running}{Running Example}{Running Examples}
\crefname{algorithm}{Algorithm}{Algorithms}
\crefname{table}{Table}{Tables}

\newcommand{\yw}[1]{#1}

\usepackage{xspace}

\newcommand{\pstl}{PSTL\xspace}
\newcommand{\hpstl}{HyperPSTL\xspace}

\usepackage{tikz}
\usetikzlibrary{automata,positioning,fit}
\usetikzlibrary{shapes,arrows}
\tikzset{state/.style={circle, draw, minimum size=0.3cm, initial distance=0.2cm}}

\usepackage{amsmath}
\usepackage{amssymb}
\usepackage{xspace}
\usepackage{latexsym}
\usepackage{mathtools}
\usepackage{relsize}
\newcommand{\midcup}{\mathsmaller{\bigcup}}

\newcommand{\phioa}{P\textsuperscript{2}HIOA}

\newcommand{\nat}{\mathbb{N}}

\newcommand{\rat}{\mathbb{Q}}
\newcommand{\real}{\mathbb{R}}

\newcommand{\nnreal}{\real_{\geq 0}}
\newcommand{\dom}{\mathrm{Dom}}
\newcommand{\img}{\mathrm{Image}}

\newcommand{\id}{\mathbf{I}}

\newcommand{\istxx}[1]{\llbracket #1 \rrbracket_{V_\sysState}}
\newcommand{\istx}[2]{\llbracket #1 \rrbracket_{V_{#2}}}

\newcommand{\bm}{\mu_\mathrm{Borel}}
\newcommand{\pr}{\mathbf{Pr}}

\newcommand{\abs}[1]{\vert #1 \vert}
\newcommand{\set}[2]{\big\{ #1 \ \vert \ #2 \big\}}

\newcommand{\ap}{\mathsf{a}} \newcommand{\aps}{\mathsf{AP}} \newcommand{\fv}{\mathsf{fv}}

\newcommand{\sys}{\mathcal{S}} \newcommand{\sysStates}{\mathcal{X}} \newcommand{\sysState}{X} \newcommand{\sysstate}{x} \newcommand{\sysinit}{x^\mathrm{init}} \newcommand{\sysInit}{X^\mathrm{init}} \newcommand{\sysin}{i} \newcommand{\sysIn}{I} \newcommand{\sysIns}{\mathcal{I}} \newcommand{\para}{d}
\newcommand{\Para}{D}
\newcommand{\Paras}{\mathcal{D}}

  \newcommand{\pv}{\pi} \newcommand{\pvs}{\Pi} \newcommand{\sig}{\sigma} \newcommand{\Sig}{\pmb{\sigma}}

\newcommand{\cdf}{\mathbf{F}}
\newcommand{\distr}{\mu}
\newcommand{\Binom}{\mathrm{Binom}}
\newcommand{\Expon}{\mathrm{Exp}}
\newcommand{\Betad}{\mathrm{Beta}}
\newcommand{\pathn}[1]{\mathrm{Path}^{#1}}

 \newcommand{\lb}{\mathtt{L}}

\newcommand{\A}{\mathcal{A}} \newcommand{\ams}{\mathtt{M_\A}} \newcommand{\ass}{\mathtt{X_\A}}   \newcommand{\ains}{\mathtt{U_\A}} \newcommand{\am}{\mathtt{m_\A}} \newcommand{\ax}{\mathtt{x_\A}} \newcommand{\au}{\mathtt{u_\A}}   \newcommand{\ainv}{\mathtt{Inv_\A}} \newcommand{\afl}{\mathtt{F_\A}} \newcommand{\atr}{\mathtt{T_\A}} \newcommand{\ai}{\mathrm{I_\A}} \newcommand{\apred}{\mathtt{Q}}

\newif\ifltl
\ltltrue
\ifltl
\renewcommand{\G}{\Box}
\newcommand{\F}{\Diamond}

\renewcommand{\U}{\mathbin{\mathcal{U}}}

\renewcommand{\P}{\mathbin{\mathbb{P}}}
\newcommand{\True}{\ensuremath{\mathtt{True}}}
\newcommand{\False}{\ensuremath{\mathtt{False}}}
\fi

\newcommand{\Ps}{ \big( \allowbreak \P^{\pvs_1} \varphi_1, \allowbreak \ldots, \allowbreak \P^{\pvs_n} \varphi_n \big) }

\newcommand{\assert}{\mathcal{A}}

\begin{document}
\title[Statistical Verification of Hyperproperties for Cyber-Physical Systems]{Statistical Verification of Hyperproperties\\for Cyber-Physical Systems}
\titlenote{This article appears as part of the ESWEEK-TECS special issue and was presented at the International Conference on Embedded Software (EMSOFT) 2019.}

\author{Yu Wang}
\orcid{0000-0002-0431-1039}
\affiliation{\institution{Duke University}
  \streetaddress{100 Science Dr, Hudson Hall Rm 220}
  \city{Durham}
  \state{NC}
  \postcode{27708}
}
\email{yu.wang094@duke.edu}

\author{Mojtaba Zarei}
\affiliation{\institution{Duke University}
  \streetaddress{100 Science Dr, Hudson Hall Rm 220}
  \city{Durham}
  \state{NC}
  \postcode{27708}
}
\email{mojtaba.zarei@duke.edu}

\author{Borzoo Bonakdarpour}
\affiliation{\institution{Iowa State University}
  \streetaddress{207 Atanasoff Hall}
  \city{Ames}
  \state{IA}
  \postcode{50011}
}
\email{borzoo@iastate.edu}

\author{Miroslav Pajic}
\affiliation{\institution{Duke University}
  \streetaddress{100 Science Dr, Hudson Hall Rm 130}
  \city{Durham}
  \state{NC}
  \postcode{27708}
}
\email{miroslav.pajic@duke.edu}

\begin{abstract}
Many important properties of cyber-physical systems (CPS) are defined upon
the relationship between multiple executions simultaneously in continuous time.
Examples include probabilistic fairness and sensitivity to modeling errors (i.e., parameters changes) for real-valued
signals.
These requirements can only be specified by {\em hyperproperties}.
In this article, we focus on verifying probabilistic hyperproperties for CPS.
To cover a wide range of modeling formalisms, we first propose a general model of probabilistic uncertain systems (PUSs) that unify commonly studied CPS models such as {\em continuous-time
Markov chains} (CTMCs) and {\em probabilistically parametrized Hybrid I/O
Automata} (\phioa).
To formally specify hyperproperties, we propose a new temporal logic,
{\em hyper probabilistic signal temporal logic} (\hpstl) that serves as a hyper and probabilistic
version of the conventional signal temporal logic (STL).
Considering the complexity of real-world systems that can be captured as PUSs, 
we adopt a {\em statistical
model checking} (SMC) approach for their verification.
We develop a new SMC technique based
on the direct computation of significance levels
of statistical assertions for \hpstl specifications,
which requires no \emph{a priori} knowledge on the indifference margin.
Then, we introduce SMC algorithms for \hpstl specifications on the joint probabilistic
distribution of multiple paths, as well as specifications with nested probabilistic operators
quantifying different paths, which cannot be handled by existing SMC algorithms.
Finally, we show the effectiveness of our SMC algorithms on CPS benchmarks with varying levels of complexity, including the Toyota Powertrain Control~System.
\end{abstract}



\maketitle

\section{Introduction} \label{sec:intro}

Ensuring safety of controllers in embedded and cyber-physical systems (CPS) 
using closed-loop system verification is a challenging problem,
due to the inherent uncertainties in system dynamics and the environment.
In systems where the uncertainties can be captured in a probabilistic manner,  two
prominent verification approaches are
{\em exhaustive}~\cite{baier_PrinciplesModelChecking_2008} and
{\em statistical}~\cite{legay_StatisticalModelChecking_2015}.
The exhaustive approach, with full knowledge of a system~model, computes the
satisfying probability of the desired properties arithmetically; 
on the other hand, the statistical approach estimates the probabilities from sampling, and
makes assertions with \yw{a certain significance level (an upper bound of the probability of returning a wrong answer)}.
Accordingly, {\em statistical model checking} (SMC) is more capable of handling
``black-box'', high-dimension or large-scale system models.

Properties of interest for such systems, from `simple' probabilistic safety or reachability, to the ones that capture dynamical responses under complex conditions, are usually formally defined by probabilistic temporal logic specifications.
Conventional probabilistic temporal logics, such as the probabilistic 
computational tree logic (PCTL)~\cite{hansson1994logic}, as well as its extension
PCTL$^*$~\cite{baier_PrinciplesModelChecking_2008},
can only specify probabilistic properties without explicitly quantifying 
over different paths of the system; that is, 
they cannot simultaneously and explicitly quantify over multiple distinctive paths to fully express their inter-relations.
This prevents them from capturing important safety/performance
hyperproperties~\cite{clarkson_Hyperproperties_2008,
abraham_HyperPCTLTemporalLogic_2018} that involve multiple execution paths, such~as \emph{sensitivity} to model errors, \emph{detectability} of system
anomalies, and \emph{fairness} when more than once process/client are
controlled/serviced.

For example, consider embedded controllers such as the Toyota Powertrain
control system benchmark~\cite{jin_PowertrainControlVerification_2014}, for
which both exhaustive (e.g.,~\cite{duggirala2015meeting})
and statistical (e.g.,~\cite{roohi_StatisticalVerificationToyota_2017}) verification techniques have been
introduced.
However, all these techniques are restricted to the use on a dynamical system model obtained
for fixed system parameters.
In general, such parameters are experimentally derived and thus, should be considered as random variables with unknown probability distributions, instead of the fixed values.
In addition, some of the system parameters might change to a degree, due to system `wear-and-tear'.
Hence, it is critical to enable analysis~of~system sensitivity to model 
errors, by providing a formal logic~to~capture such properties, as well as 
methods to verify how system evolution (execution) changes for different 
parameters of the system~model.

Specifying these properties, such as sensitivity to  change of parameters, involves probabilistic
quantification over multiple paths (a path and its deviation), and thus can only
be captured as {\em hyperproperties}~\cite{clarkson_Hyperproperties_2008}.
For example, as illustrated in \cref{fig:sensitivity}, for sensitivity analysis we can check whether the deviation $\pi_2$ of a path $\pi_1$ under probabilistic uncertainty stays close probabilistically, such that there is limited variation in the hitting time $\tau$ to a desired working region.
Although the sensitivity may be expressed as a non-hyperproperty if `expected' hitting times~are~known in advance (as a reference) for any entry into~the~desired operating~region, such information is usually~unavailable for complex~systems.

Consequently, in this work, we first introduce a probabilistic temporal logic
for hyperproperties expressed on real-valued continuous-time signals, referred to as {\em Hyper
Probabilistic Signal Temporal Logic} (\hpstl).
\hpstl can be viewed as 
a hyper extension and generalization of the probabilistic signal/metric-interval temporal logic~\cite{sadigh_SafeControlUncertainty_2016,
wang_VerifyingContinuoustimeStochastic_2016}, 
a probabilistic version of HyperSTL~\cite{nguyen_HyperpropertiesRealvaluedSignals_2017},
\yw{and a continuous-time extension and generalization of 
HyperPCTL~\cite{abraham_HyperPCTLTemporalLogic_2018}.}
\hpstl extends those logics by enabling
(1) reasoning about the probability of paths by adding a probability operator,
and (2) reasoning about multiple paths simultaneously, i.e., hyperproperties
specifying the relationship between different paths by associating path
variables to the atomic propositions.

To allow us to cover a range of modeling formalisms, we introduce
a very general system model -- {\em probabilistic uncertain systems} (PUS) --
and define the semantics of \hpstl on it.
\yw{Generally, they are `black-box' probabilistic dynamical systems with unknown dynamics on a given state space.}
A PUS incorporates nondeterminism as its input and probabilism as its parameters, both of which are time functions of values of general types, including real, integer or categorical/Boolean.
Given the values of the input and the parameters, \yw{we can draw}
a time-dependent sample path \yw{from the PUS}, which can also be of general types.
We note that this general model captures commonly studied models such as
CTMCs and hybrid I/O automata with probabilistic parameters -- referred to \yw{as} probabilistically
parameterized hybrid I/O automata (\phioa).

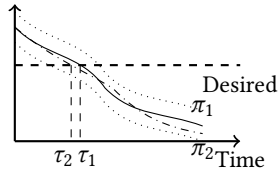
\begin{figure}[!t]
\centering
\begin{tikzpicture}
    \draw [->, thick] (0, 0) -- (0, 1.8);
    \draw [->, thick] (0, 0) -- (3, 0) node[below] () {\small Time};
    \draw [dashed, thick] (0, 1) -- (3, 1) node[below] () {\small Desired};
	\draw [] (0,1.5) to[out=-45,in=125] (1.2,0.7) to[out=-45,in=160] (2.5,0.2) node[above] () {$\pv_1$};
	\draw [yshift=0.2cm, dotted] (0,1.5) to[out=-45,in=125] (1.2,0.7) to[out=-45,in=160] (2.5,0.2);
	\draw [yshift=-0.2cm, dotted] (0,1.5) to[out=-45,in=125] (1.2,0.7) to[out=-45,in=160] (2.5,0.2);
	\draw [dash dot] (0,1.5) to[out=-40,in=150] (1.3,0.7) to[out=-45,in=170] (2.5,0.1) node[below] () {$\pv_2$};

    \draw [dashed] (0.87, 1) -- (0.87, 0) node[below, xshift=0.1cm] () {$\tau_1$};
    \draw [dashed] (0.75, 1) -- (0.75, 0) node[below, xshift=-0.1cm] () {$\tau_2$};
\end{tikzpicture}
\caption{Sensitivity of paths $\pv_1$ and $\pv_2$, drawn in \yw{solid and dashed} lines respectively.\label{fig:sensitivity}}
\end{figure}

\yw{We define the semantics of \hpstl on a set of paths of the PUS through pre-defined labels or predicates on the state spaces.}
We show \yw{by concrete examples} that this allows for the capturing of properties such as (i) anomaly 
detectability, or (ii) sensitivity to model errors due to probabilistic uncertainty of parameters of hybrid I/O automata, and (iii) workload~fairness in queueing networks modeled as continuous-time
Markov chains (CTMC).
\yw{To verify \hpstl specifications on a PUS for a given input, we introduce 
new SMC algorithms using Clopper-Pearson (CP) bounds~\cite{clopper1934use}}.

Unlike most previous SMC methods based on the sequential probability {ratio} test (SPRT)~\cite{legay_StatisticalModelChecking_2015,larsen_StatisticalModelChecking_2016}, this approach requires no \emph{a priori} knowledge on the indifference margin.
\yw{Our SMC algorithm can verify a \hpstl specification to arbitrarily small (non-zero) significance levels, which is also different from~\cite{barbot_StatisticalModelCheckingAutonomous_2017} on using CP bounds to estimate satisfying probability for given~samples.} 
\yw{The conservativeness of the CP significance level ensures that the desired significance level is strictly achieved even in the worst case.
Specifically, we iteratively draw samples from the objective model and compute the significance level derived from the CP confidence interval.
The algorithm stops if the CP significance level is smaller than the desired significance levels.}
Further improvement in sample efficiency is also possible using the upper bounds derived from Wilson scores, Jeffreys interval or Agresti-Coull interval, at the cost of yielding only asymptotic correctness (see~\cite{brown_IntervalEstimationBinomial_2001} for a review of these confidence~intervals).

\yw{To the best of our knowledge, this work is the first to enable} SMC on probabilistic temporal
logic for hyperproperties in continuous time.
Compared to common non-hyper temporal logics, \hpstl allows uniquely for (1)
probability quantification over multiple paths, (2) defining specifications on
their joint probabilistic distribution, such as comparison of probabilities, and
(3) nesting of probabilistic operators that quantify over different sets of
paths.
We address these challenges by deriving SMC methods capable of handling these cases, while providing provable significance levels.
We also show that \hpstl subsumes several existing signal temporal
logics for probabilistic
properties~\cite{roohi_StatisticalVerificationToyota_2017,
wang_VerifyingContinuoustimeStochastic_2016} (see \cref{sec:logic}), and
thus, the SMC algorithms for \hpstl  introduced in this work,
also apply to them.

Finally, we show the effectiveness of our~SMC methods on embedded
case studies with different complexity and modeling formalism.
Specifically, we statistically verified sensitivity of a thermostat and the
Toyota Powertrain System with uncertain parameters, as well as fairness in queueing networks
of different sizes; we generated probabilistic guarantees with high
significance levels. For example, for the Toyota Powertrain and the large queueing networks, where 
exhaustive verification is not possible, we construct upper bounds on the 
$0.95$ and $0.99$ percentiles of their sensitivity and fairness, respectively, 
with a confidence level of $0.99$ using only a few hundred samples.
This achieves, for the first time, statistical verification of probabilistic
hyperproperties on continuous signals, which can be used even for real-world~sized~CPS.

\subsubsection*{Organization} We introduce probabilistic uncertain systems in
\cref{sec:model},
followed by the syntax and semantics of \hpstl~in~\cref{sec:logic}. In
\cref{sec:example}, we show how  \hpstl can be used to capture  desired
properties of CPS.
Our SMC algorithms are presented in \cref{sec:smc} with an emphasis on the
use of CP significance levels and the handling of probabilities involving
multiple paths.
We apply the SMC techniques to three embedded case studies (e.g., Toyota Powertrain System) in \cref{sec:evaluation}, before concluding~in~\cref{sec:conc}.

\subsubsection*{Notation}
We denote the sets of integers, rational, real, and non-negative real numbers by $\nat$, $\rat$, $\real$, and $\nnreal$, respectively.
The domain and image of a function is denoted by $\dom(\cdot)$ and $\img(\cdot)$, respectively.
Let $\nat_\infty = \nat \cup \{\infty\}$ and $\rat_\infty = \rat \cup \{\infty\}$.
For $n \in \nat$, let $[n] = \{1,\ldots,n\}$.
The indicator function is denoted by $\id$.
We denote the (Borel) measure of a measurable set by $\bm (\cdot)$.
The cardinality and the power set of a set are denoted by $\abs{\cdot}$ and $2^{\cdot}$.
For any set $I \subseteq \real^n$, we denote its boundary, interior, and 
closure by $\partial I$, $I^\circ$, and $\bar{I}$, respectively.
The empty set is denoted by $\emptyset$.
With a slight abuse of notation, given a map $V: A \rightarrow B$ and $A' 
\subseteq A$, let $V(A') = \midcup_{a \in A'} V(a)$.

We refer to a function of time $\sig: \nnreal \rightarrow \real^n$ as a \emph{signal}, and denote by $\sig^{(t_1)}$ its $t_1$-time shift defined as $\sig^{(t_1)}(t_2) = \sig(t_1 + t_2)$ for $t_1 \in \nnreal$.
We denote the binomial distribution by $\Binom(n, p)$, the exponential distribution parametrized by rate by $\Expon(\lambda)$, and the beta distributions with shape parameters by $\Betad(\alpha, \beta)$.
A random variable $X$ drawn from probability distribution $\distr$ is denoted by~$X\sim\mu$.

\section{Probabilistic Uncertain Systems}
\label{sec:model}

To allow for defining \hpstl for a large class of commonly used modeling formalisms in a way that facilitates design of SMC techniques,
we introduce a very general system model, which we refer to as 
{\em probabilistic uncertain systems} (PUS) (\cref{fig:probabilistic_system}). 
Such model can be viewed as a ``black-box'' probabilistic dynamical system with an explicit
state space, where the randomness only comes from the time-dependent parameters
drawn from random processes.
The PUS generalizes models such as CTMCs and \phioa, and provides a unified framework for the SMC of probabilistic hyperproperties.

A PUS is a 'black-box' probabilistic dynamical system on a given {\em state space}, 
which we denote by $\sysStates$. 
Its probabilistic uncertainty comes from a set of $n$ \emph{time-dependent} parameters $\Para(t)
= \big( \para_1(t), ..., \para_n(t) \big)$ where $t \in \nnreal$, drawn from some probability
distribution $\distr(\Paras)$ on its domain $\Paras$.
That is, the {\em parameters} of the system are drawn from an $n$-dimensional
random process.
The {\em input} to the system $\sysIn(t) = \big( \sysin_1(t), \dots,
\sysin_m(t) \big) \in \sysIns$ is an $m$-dimensional function of time $t$.
Given the input $\sysIn(t)$, the values of the parameters $\Para(t)$, and an initial 
state $\sysInit = \big( \sysinit_1, \dots, \sysinit_l \big) \in \sysStates$,
the system deterministically generates
a time-dependent {\em path} $\sysState: \nnreal \rightarrow \sysStates$ with $\sysState(t) = \big( \sysstate_1(t), \dots, \sysstate_l(t) \big)$. 
That is, the randomness in system evolution only comes 
from system~parameters.

The values of the system inputs, parameters, and states can be of {mixed}
types: real, discrete or categorical, depending on the system formulation.
Here, we make no assumption on the system dynamics (Markovian, causal, etc); rather, the system should be viewed just as a general
deterministic map from the functions $\sysIn(t)$ and $\Para(t)$ to the function
$\sysState(t)$.
Without knowing the value of the system parameters and given the initial state,
the map from the input $\sysIn(t)$ to the path $\sysState(t)$ is {only
probabilistic, \yw{i.e., the probability distribution for all uncertain variables are given.}

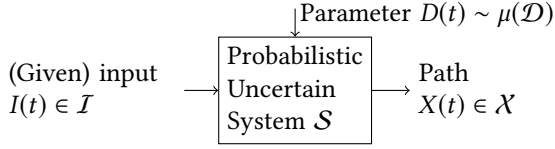
\begin{figure}[!t]
\centering
\begin{tikzpicture}
	\node[text width=2.2cm] (in) {\yw{(Given)} input $\sysIn(t) \in \sysIns$};
  	\node[draw, right of=in, xshift=0.4cm, minimum width=2cm, minimum height=1cm, xshift=1.3cm, text width=1.8cm] (sys) {Probabilistic Uncertain System $\sys$};
  	\node[text width=1.7cm, above of=sys, yshift = 0.2cm] (para) {};
	\node[text width=3.7cm, above of=sys, xshift = 1.92cm] (paraT) {Parameter $\Para(t) \sim \distr(\Paras)$};
  	\node[text width=1.3cm, right of=sys, xshift = 1.3cm] (out) {Path $\sysState(t) \in \sysStates$};
	\path[->] (in) edge (sys.west);
	\path[->] (para) edge (sys.north);
	\path[->] (sys.east) edge (out);
\end{tikzpicture}
\caption{Probabilistic Uncertain System $\sys$.\label{fig:probabilistic_system}}
\end{figure}

Given the finite set of {\em atomic propositions} $\aps$,
the labeling function on the state space
$\lb:
\sysStates \rightarrow 2^\aps$
defines for each state, a set of atomic propositions that hold.
Alternatively, the labeling of each atomic proposition in $\aps$ can be represented by a predicate/Boolean function that indicates the subset of states $\sysStates$ that are labeled by that atomic proposition.
Then, a path of the system induces a signal $\sig(t) = \lb(\sysState(t)): \nnreal
\rightarrow 2^\aps$.
This signal indicates the set of properties that is satisfied by the system
$\sys$ at each time $t \in \nnreal$ during system evolution.
In \cref{sec:logic}, we will define hyper temporal properties on sets of such signals.

A PUS $\sys$ is both probabilistic in the value of its parameters and
nondeterministic in its input.
We note that the distinction between the parameters and the input of the system
is only mathematical, not physical.
In practice, we can model the probabilistic input of a real system as a
parameter of its PUS model, and a nondeterministic parameter of the system as an
input of the PUS.
In this work, we focus on SMC related to the probabilistic nature~of~a~PUS;
 i.e., we consider the cases, where the input of the PUS is given such as
when it is determined by the actions of the controller used to control the
system (PUS).
Before we introduce \hpstl and SMC methods for PUS, in the rest of this
section, we show how our abstract notion of PUS captures
two prominent computing models: \emph{continuous-time Markov chains}
(CTMCs) that mathematically model queuing networks and probabilistic hybrid I/O
automata that model the Toyota Powertrain system {mentioned in}~\cref{sec:intro}.

\subsection{Probabilistically Parameterized Hybrid I/O Automata}
\label{sub:hybrid_automata}

Probabilistically Parameterized Hybrid I/O Automata (\phioa) are extensions of
Hybrid I/O Automata (HIOA)~\cite{lynch2003hybrid,henzinger2000theory,jin_PowertrainControlVerification_2014} with probabilistic system
parameters.
Given the variety of mathematical models for HIOA, we build the one used
in~\cite{jin_PowertrainControlVerification_2014} to describe the dynamics
of hybrid systems with fully observable states.
Specifically, an HIOA is a tuple $\A = (\ams, \ass, \ai, \ains, \allowbreak \ainv, \atr, \apred)$ where
\begin{itemize}
\item $\ams$ is a finite set of {\em modes};

\item $\ass $ is a set of $n$ {\em state variables} of real values, i.e.,
$\ass\subseteq \real^n$;

\item $\ains$ is a set of $m$ {\em input variables}; their valuations can
be of different types, such as $\nat$, $\real$, and Boolean;

\item $\afl: \ams \times \ass \times \ains \rightarrow
\real^{n}$ defines a deterministic {\em flow} capturing system evolution in
each mode, i.e., a differential equation:
$$\frac{\mathrm{d} \ax}{\mathrm{d} t} = \afl\big(\am, 
\ax, \au\big);$$

\item For any mode
$\am \in \ams$, $\ainv(\am) = \dom(\afl(\am, \cdot, \cdot))$ defines the
{\em invariant} of the mode;

\item $\atr: \ams \times \ass \times \ains \rightarrow \ams \times
\ass$ defines deterministic {\em jumps} triggered by $\ass \in \partial \ainv(\am)$,
i.e., when the flow hits the boundary of the invariant.

\item $\ai \in \ams \times \ass$ is the {\em initial} condition;

\item $\apred \subseteq \ains \times \ass \rightarrow \{0,1\}$ is a finite set
of {\em predicates}.

\end{itemize}

The \phioa~extend Hybrid I/O Automata by allowing system parameters, used to capture $\afl, \atr, \ainv$, to be probabilistic instead of fixed. 
This differs from the Probabilistic Hybrid Automata introduced in~\cite{sproston_DecidableModelChecking_2000,zhang_SafetyVerificationProbabilistic_2010}, where the randomness comes only from the probabilistic jumps.
In standard hybrid models,
the deterministic and nonlinear flow, invariant, and jump functions $\afl, \atr, \ainv$, which define
its dynamics, are parameterized by quantities that are estimated from physical
experiments.
For example, in the Toyota Powertrain model, the mass flow rate of intake air is determined by the RPM of the
engine and the pressure of the intake manifold through a polynomial, whose five
parameters are fit from experimental data.
Due to experimental errors, these parameters are \yw{better}
represented as random variables with unknown
probability distributions (e.g.,
Gaussian or uniform
with means and variances inferred from experimental
data) than real numbers as in~\cite{jin_PowertrainControlVerification_2014}.
Thus, we denote the parameters for $\afl$, $\atr$, and $\ainv$ by $\Para = (\para_\afl,
\para_\atr, \para_\ainv)$.
To simplify our presentation, in this work, we assume that these parameters for \phioa~are time-invariant.

Consequently, a \phioa~can be represented by a PUS by treating (i) $\sysIn = \ains$
as the input, (ii) \yw{$\sysState = \ams \times \ass$} as
the state (the input variables $\ains$ are encoded as part of the state), (iii) $\sysInit = \ai$ as the initial state, (iv) $\Para = (\para_\afl, \para_\atr,
\para_\ainv)$ as the parameters, and (v)~the predicate $\apred$ as the labeling
function.
The probabilistic hyperproperties related to \phioa,  as the one discussed in \cref{sec:intro,sec:example}, will be statistically verified on this~PUS.

\subsection{Continuous-Time Markov Chains} \label{sub:ctmc}

Another example of PUS are CTMCs, which are
commonly used to model queuing and task scheduling in embedded computing and communication systems with
uncertainties.
Consider a CTMC~with
\begin{itemize}
	\item the {\em states} $[n]$,
	\item the {\em initial state} $\sysState_0 \in [n]$,
	\item the {\em labeling function} on the states $\lb: [n] \rightarrow
2^\aps$ for a given set of labels $\aps$,
	\item the {\em probability transition rate matrix} $M \in \real^{n
\times n}$, such that $\sum_{j \in [n]} \allowbreak M_{ij} = 0$, where $M_{ij}$
is the transition rate from a state $i$ to a state $j$.
\end{itemize}

The CTMC can be represented by a PUS with (i) the states $\sysStates = [n]$, (ii) the initial state $\sysState_0$, (iii) the empty inputs  (i.e., the PUS has no input), (iv) the labels $\aps$, and (v) the labeling function $\lb: [n] \rightarrow 2^\aps$.
For $i \in \nat$, we draw $\para_{1i}$ and $\para_{2i}$ from $[0,1]$ uniformly and independently.
Using Gillespie algorithm~\cite{gillespie_GeneralMethodNumerically_1976}, 
with the convention $\sum_{k=1}^{0} \cdot = 0$, 
the state $\sysState$ of the PUS representing the CTMC evolves~by:
\[
	\sysState(t) = \sysState_i ~~~~~~~~~~ \textrm{ if }~~~~~~~~ t \in [T_i, 
T_{i+1}),
\]
where the state jumps are determined by $\para_{1i}$ as
\[
	\sysState_{i+1} = j \textrm{ when } \sysState_{i} = l \ \textrm{ if } \ \sum\nolimits_{k=1}^{j-1} M_{l
k} \leq \para_{1i} < \sum\nolimits_{k=1}^{j} M_{l k},
\]
and the time lapses are determined by $\para_{2i}$ using
\[
T_{i+1} = T_i + \cdf^{-1}_{\Expon \big( \sum_{k \in [n], k \neq l} M_{l k} \big) } \big( \para_{2i} \big), \ T_0 = 0,
\]
with $\Expon(\cdot)$ denoting the exponential distribution parameterized by
rate and $\cdf^{-1}_{\cdot}$ the inverse function of the cumulative
distribution function.
Treating the two sequences $\{\para_{1i}\}_{i \in \nat}$ and $\{\para_{2i}\}_{i \in \nat}$ as the two parameters of the PUS, provides a PUS presentation of the CTMC.
Although they are discrete sequences, we can easily represent them as continuous functions of time to fit into the definition of the parameters of the PUS.
Then, the probabilistic hyperproperties (e.g., fairness) for the initial CTMC, should be
statistically verified on the aforementioned PUS.

\section{Hyper Probabilistic Signal-Temporal Logic}
\label{sec:logic}

To formally express and reason about probabilistic hyperproperties on
real-valued signals, we introduce the logic \hpstl, which can be viewed as a probabilistic extension
of the {\em signal temporal logic} (STL) for hyperproperties.
We introduce the syntax and semantics of \hpstl for PUS in~\cref{sub:syntax,sub:semantics}, before presenting its use to capture relevant properties of embedded systems in \cref{sec:example}.

\subsection{Syntax} \label{sub:syntax}
We define \hpstl
formulas inductively as:
\begin{align}
  & \varphi \Coloneqq \ap^\pv
    \ \vert \ \varphi^\pv
		\ \vert \neg \varphi
		\ \vert \ \varphi \land \varphi
		\ \vert \ \varphi \U_{[t_1, t_2]} \varphi
		\ \vert \ p \Join p \label{eq:hpstls_1}
	\\ & p \Coloneqq \P^{\pvs} \varphi
		\ \vert \ \P^{\pvs} p
		\ \vert \ f(p, \dots, p) \label{eq:hpstls_2}
\end{align}
where
\begin{itemize}
  \item $\ap\in\aps$, and $\aps$ is the finite set of {\em atomic propositions},
  \item $t_1 < t_2$ with $t_1, t_2 \in \rat_\infty$,
  \item $\pv$ is a path variable, and
$\pvs$ is a set of path variables,
  \item $\P$ is the probability operator,
  \item $\Join \ \in \{<, >,=,\leq,\geq\}$,
  \item $f: \real^n \rightarrow \real$ is a $n$-ary measurable
function, constants are viewed as $0$-ary functions,
  \item $\fv(\cdot)$ denotes the set of free path variables in
$\varphi$, \yw{i.e.,} the path variables not quantified by a probability operator
through~\eqref{eq:hpstls_2}, and $\fv(\varphi) = \emptyset$ in the
$2^\mathrm{nd}$ rule of~\eqref{eq:hpstls_1}. This is recursively defined
 by:
  \[
  \begin{split}
    & \fv(\ap^\pv) = \{ \pv \}, \ \fv(\varphi^\pv) = \{ \pv \}
  , \ \fv(\neg \varphi) = \fv(\varphi),
    \\ & \fv(\varphi_1 \land \varphi_2) = \fv(\varphi_1) \cup \fv(\varphi_2)
  , \ \fv(\varphi_1 \U_{[t_1, t_2]} \varphi_2) = \fv(\varphi_1) \cup \fv(\varphi_2),
    \\ & \fv(p_1 \Join p_2) = \fv(p_1) \cup \fv(p_2), \ \fv(f(p_1, \ldots, p_n)) = \midcup_{i \in [n]} \fv(p_i),
    \\ & \fv(\P^{\pvs} (\varphi)) = \fv(\varphi) \backslash \pvs, \
\fv(\P^{\pvs} (p)) = \fv(p) \backslash \pvs.
\end{split}
  \]
\end{itemize}

\yw{Other common logic operators can be derived as follows:}
$\varphi \lor \varphi' \equiv \neg (\neg \varphi \land \neg \varphi')$,
$\True \equiv \varphi \vee \neg \varphi$, $\varphi
\Rightarrow \varphi' \equiv \neg \varphi \lor \varphi'$, $\F_{[t_1, t_2]} \varphi
\equiv \True \, \U_{[t_1, t_2]} \, \varphi$, and $\G_{[t_1, t_2]} \varphi \equiv \neg
\F_{[t_1, t_2]} \neg \varphi$.
\yw{In addition, we denote $\U_{[0, \infty)}$, $\F_{[0, \infty)}$, and $\G_{[0, \infty)}$ by $\U$, $\F$, $\G$, respectively.}

A formula with $\fv(\varphi) = \emptyset$ is called a {\em state formula} and 
requires no instantiation of free path variables, thus can be evaluated on a 
state of the PUS.
Therefore, we can associate another path variable $\pv$ to it.
All other formulas are referred to as \emph{path formulas}, as their 
correctness depends on the instantiation of their free path variables.
It specifies on which path a state formula should be satisfied, so that a \hpstl formula can reason simultaneously on multiple paths.
Finally, for the first two rules of \eqref{eq:hpstls_2}, we assume that $\pvs$
is contained in $\fv(\varphi)$ or $\fv(p)$ for the probability quantification
to be non-trivial.
Observe that \hpstl can be viewed as the probabilistic version of
HyperSTL~\cite{nguyen_HyperpropertiesRealvaluedSignals_2017} by replacing the
existential and universal quantifiers over signals with probabilistic
quantifiers over paths in~\eqref{eq:hpstls_2}.

\hpstl has the following unique features.
It allows for the simultaneous probability quantification over several paths
(as we show for the sensitivity analysis of powertrain controllers in
\cref{sec:example}).
It also allows for the arithmetics and comparison of probabilities, and the nesting
of probability operators quantifying different paths (as shown in the queueing
fairness analysis described in \cref{sec:example}).

\hpstl reduces to a non-hyper probabilistic signal temporal logic (PSTL)
if it only has one path variable in it.
PSTL can still define probability satisfaction of single atomic 
propositions, so it subsumes the MITL for probability distributions
from~\cite{wang_VerifyingContinuoustimeStochastic_2016}.
But, PSTL (and thus \hpstl) does not subsume PrSTL
in~\cite{sadigh_SafeControlUncertainty_2016}, since it does not allow a
time-varying probability threshold, as is allowed in PrSTL.
Still, augmenting \hpstl syntax to allow time-varying functions is
straight-forward.

Finally, note that to simplify our presentation of \hpstl syntax and  
semantics, while allowing for verification of complex systems such as \phioa, we 
only include simultaneous or consecutive probabilistic quantification (e.g., 
$\P^{\{\pv_1, \pv_2\}}$ or $\P^{\pv_1} \P^{\pv_2}$) over the paths from a single \yw{initial}
state.
\hpstl can be augmented by allowing nested existential and universal
quantification over multiple states in the same way
as~\cite{wang2019statistical,abraham_HyperPCTLTemporalLogic_2018}.
Specifically, in addition to the probabilistic quantification over the paths, 
one can add
extra state quantification of these paths to specify from which state the path
starts, like $\exists \sysState_1^{\pv_1}. \forall \sysState_2^{\pv_2}.
\P^{\{\pv_1, \pv_2\}}$.
However, verifying such formulas generally requires exhaustive iteration over
all the states, which is challenging, if not impossible, on systems with
infinite state spaces like \phioa. 
Therefore, in this paper, we do not include state quantification which in our logic can be done as
presented in~\cite{wang2019statistical,abraham_HyperPCTLTemporalLogic_2018}.

\subsection{Semantics} \label{sub:semantics}

We define the satisfaction relation for \hpstl state formulas~on a PUS
$\sys$ by
\begin{equation} \label{eq:semantic1}
  \begin{array}{l@{\hspace{1em}}c@{\hspace{1em}}l}
(\sys, \sysState) \models \varphi & \Leftrightarrow &  \sys \models \istx{\varphi}{\sysState},
\end{array}
\end{equation}where $V_\sysState$ is an assignment of path variables to the paths of the PUS $\sys$
starting from a state $\sysState$, and $\istxx{\varphi}$ is the instantiation of
the assignment $V_\sysState$ on $\varphi$.
Here, the instantiation specifies the initial state for the formula $\varphi$.

The satisfaction relation for the \hpstl path formulas is defined with respect to the assignment
$V_\sysState$ by:
\begin{equation} \label{eq:semantic2}
\begin{array}{l@{\hspace{1em}}c@{\hspace{1em}}l}
\sys \models \istx{p \Join p}{\sysState} & \Leftrightarrow & \sys \models
\istx{p}{\sysState} \Join \istx{p}{\sysState}
\\ \sys \models \istx{f(p, \ldots, p)}{\sysState} & \Leftrightarrow &
\sys \models f \big(\istx{p}{\sysState}, \ldots, \istx{p}{\sysState}\big)
\\ \sys \models \istx{\P^{\pvs} (\varphi)}{X} & \Leftrightarrow &
\sys \models \pr_{\Sig \sim \pathn{\abs{\pvs}} (\sysState)} \Big( \big( \sys, V_\sysState[\pvs \rightarrow \Sig] \big) \models \varphi \Big)
\\ (\sys, V_\sysState) \models \ap^\pv & \Leftrightarrow & \ap \in \lb(V_\sysState(\pv) (0))
\\ (\sys, V_\sysState) \models \Phi^\pv & \Leftrightarrow & (\sys, V_\sysState(\pv)) \models \Phi
\\ (\sys, V_\sysState) \models \neg \varphi & \Leftrightarrow & (\sys, V_\sysState) \not\models \varphi\\
(\sys, V_\sysState) \models \varphi_1 \land \varphi_2 &  \Leftrightarrow & (\sys, V_\sysState) \models
\varphi_1 \textrm{ and } (\sys, V_\sysState) \models \varphi_2
\\ (\sys, V_\sysState) \models \varphi_1 \U_{[t_1, t_2]} \varphi_2  & \Leftrightarrow & \exists t \in [t_1, t_2]. \ \Big(\forall t' < t. \big( \sys, V_\sysState^{(t')} \big) \models \varphi_1 \Big) \land \big( \sys, V_\sysState^{(t)} \big) \models \varphi_2
\end{array}
\end{equation}
where
\begin{itemize}
  \item $\lb(\sysState)$ is the set of labels of the PUS state $\sysState$,
  \item $\pathn{\abs{\pvs}} (\sysState)$ is the collection of all
$\abs{\pvs}$-tuples of paths starting from state $\sysState$,
  \item $V_\sysState[\pvs \rightarrow \Sig]$ is a revision of the assignment
$V_\sysState$ by assigning the set $\Sig$ of paths to the set $\pvs$ of path 
variables, respectively,
  \item $V_\sysState^{(t)}$ is the $t$-shift of the assignment $V_\sysState$, 
defined by $\big(V_\sysState^{(t)} (\pv)\big) \allowbreak = 
(V_\sysState(\pv))^{(t)}$ for all path variables $\pi$ in assignment~$V$.
\end{itemize}
Finally, we note the equivalence $\ap \in \lb(V_\sysState(\pv) (0)) \  
\Leftrightarrow \ \ap \in \lb(\sysState)$ and $(\sys, V_\sysState(\pv)) \models 
\Phi \ \Leftrightarrow \ (\sys, \sysState) \models \Phi$ for the $4^\mathrm{th}$ 
and $5^\mathrm{th}$ rules in~\eqref{eq:semantic2}. The following example 
illustrates the semantics of \hpstl.

\begin{example} \label{ex:conditional_probability}
As shown in \cref{fig:example1}, consider a CTMC $\sys$ that models a queue 
that is initially empty and is of buffer size $2$. 
The transition rate matrix of $\sys$ is the following:
$$\begin{bmatrix}
  - 1 & 1 & 0 \\
  2 & -3 & 1 \\
  0 & 2 & -2
\end{bmatrix}$$
The CTMC satisfies the following
\hpstl formula:
\[
\varphi = \P^{\pv_1} \big( (\neg s_1^{\pv_1}) \U ( s_1^{\pv_1} \U_{[0, 1]} s_0^{\pv_1} ) \big) - \P^{\pv_2} \big( (\neg s_2^{\pv_2}) \U ( s_2^{\pv_2} \U_{[0, 1]} s_1^{\pv_2} ) \big) > 0.05,
\]
The formula asserts that the probability difference between (i) finishing a task within $1$ time delay after first having $1$ in queue
and (ii) finishing a task within $1$ time delay after first having $2$ in queue, is greater than $0.05$.
The is because the probability for (i) is $\frac{2}{3} (1 - e^{-3}) \approx 0.633$ and that for (ii) is $1 - e^{-2} \approx 0.865$.
\hfill $\lhd$
\end{example}

\begin{figure}[!t]
\centering
\begin{tikzpicture}
  \node[state, initial text=, initial] (0) {$s_0$};
  \node[state, right of=0, xshift=0.5cm] (1) {$s_1$};
  \node[state, right of=1, xshift=0.5cm] (2) {$s_2$};

  \path[->] (0) edge[bend left] node[above] {$1$} (1);
  \path[->] (1) edge[bend left] node[above] {$1$} (2);

  \path[->] (2) edge[bend left] node[above] {$2$} (1);
  \path[->] (1) edge[bend left] node[above] {$2$} (0);

\end{tikzpicture}
\caption{CTMC Model of A Queue \hpstl.\label{fig:example1}}
\end{figure}
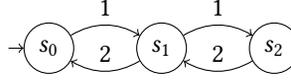

\subsection{Expressivity of \hpstl}

\begin{theorem}
\hpstl subsumes \pstl on CTMCs.
\end{theorem}

\begin{proof}

It suffices to show that there exist formulas in \hpstl that cannot be expressed in \pstl.
The general idea is to show that \pstl cannot express conditional probabilities, while \hpstl can.
Consider the CTMC $\sys$ given in \cref{fig:thm1}, and the \hpstl state formula
\[
\varphi = \Big( \frac{\P^\pv\big( \mathrm{Init}^\pv \Rightarrow \F  (\ap_1^\pv \land \ap_2^\pv)\big)}{\P^\pv\big( \mathrm{Init}^\pv \Rightarrow \F \ap_2^\pv\big)} = \frac{1}{2} \Big),
\]
expressing a conditional probability for a path $\pi$ of $\sys$.
Clearly, it is satisfied for the state $\mathrm{Init}$.

\begin{figure}[!t]
\centering
\begin{tikzpicture}
  \node[state, initial text=, label={0:$\mathrm{Init}$}] (0) {$\sysState_0$};
  \node[state, below of=0, label={270:$\{\ap_1\}$}, yshift = 
-.5cm, xshift=-1.5cm] (1) {$\sysState_1$};
  \node[state, below of=0, label={270:$\{\ap_2\}$}, yshift = 
-.5cm] (2) {$\sysState_2$};
  \node[state, below of=0, label={270:$\{\ap_1,\ap_2\}$}, xshift=1.5cm, yshift 
= -.5cm] (3) 
{$\sysState_3$};

  \path[->] (0) edge node[right] {$1$} (1)
  (0) edge node[right] {$1$} (2)
  (0) edge node[right] {$1$} (3);
\end{tikzpicture}
\caption{\hpstl on CTMC.\label{fig:thm1}}
\end{figure}

We claim that $\varphi$ cannot be expressed in \pstl.
By the syntax and semantics of \pstl~\cite{sadigh_SafeControlUncertainty_2016}, it suffices to show that $\varphi$ cannot be expressed by a formula $\P (\psi)$, where $\psi$ is a \pstl path formula derived by concatenating  a set of \pstl state formulas $\varphi_1, \ldots, \varphi_n$ with $\land, \neg$, or temporal operators.
These state formulas are either $\True$ or $\False$ on the states $\sysState_0$, $\sysState_1$, $\sysState_2$, and $\sysState_3$.
Thus, whether a path satisfies $\psi$, defines a subset of the paths
from the state $\sysState_0$ in the CTMC,
Since the probability of a path ended up in any $\sysState_i$ is $1/3$ for $i = 1,2,3$, the formula $\P (\psi)$ can only take values in $\{0, 1/3, 2/3, 1\}$.
However, by the semantics of \hpstl, the fractional probability on the right side of the \yw{equation} has value $1/2$, thus $\varphi$ evaluates to true and cannot be expressed by $\P (\psi)$~in~\pstl.
\end{proof}

\section{H\lowercase{yper}PSTL in Action}
\label{sec:example}

In this section, we demonstrate how \hpstl can be used to capture relevant 
properties of CPS.

\subsection{Sensitivity to Modeling Errors in \phioa} 

A typical example of a CPS that can be modeled as a \phioa~is the automotive powertrain, where the response to the change in driving behaviors is of key interest.
Note that the dynamical system parameters, even for the same type of powertrains, vary across different systems. Thus, it is critical to analyze if \emph{in most cases},~the~change in dynamical response of the controlled system stays within permitted amount $\delta$ when the system parameters change; i.e., if  the dynamical response deviation is within $\delta$ with probability of at least $1 - \varepsilon$.

Consider as an example, the sensitivity of the Toyota Powertrain 
Controller~\cite{jin_PowertrainControlVerification_2014}  under probabilistic 
uncertainty in its dynamical~parameters (i.e., the system model). As shown in \cref{fig:sensitivity}, we consider statistically verifying the 
probabilistic boundedness of the sensitivity of the (first) hitting time 
$\tau$ to a desired working region (where the error to the desired A\/F ratio is 
less than $5\%$) under the probabilistic uncertainty in RPM.
Mathematically, this can be represented~by
\[
\pr_{\pv_1, \pv_2} \big(\abs{\tau^{\pv_1} - \tau^{\pv_2}} \leq \delta \big) > 1 - \varepsilon,
\]
for some given values $\delta, \varepsilon > 0$, where $\pv_1$ and $\pv_2$ are two \yw{statistically independent} sample paths of the system.
To express this formally using \hpstl,~we introduce a predicate $\apred$ for the desired working region of the system (i.e., under the dashed line in \cref{fig:sensitivity}).
If $0 \leq \tau^{\pv_2} - \tau^{\pv_1} \leq \delta$, the switch from $\neg \apred^{\pv_2}$ to $\apred^{\pv_2}$ for the path $\pv_2$ happens within time $\delta$ after $\neg \apred^{\pv_1}$ changes to $\apred^{\pv_1}$ for the path $\pv_1$.\footnote{Meanwhile, $\apred^{\pv_2}$ may switch back to $\neg \apred^{\pv_2}$ for $\pi_2$, but the first hitting time $\tau^{\pv_2}$ will not change.} 
This can be equivalently expressed as $(\neg \apred^{\pv_1} \land \neg \apred^{\pv_2}) \U (\apred^{\pv_1} \yw{\land \F_{[0, \delta]}} \apred^{\pv_2})$, and accordingly, the probabilistic boundedness of sensitivity is expressed as
\begin{equation} \label{eq:spec_tp}
\P^{\{\pv_1, \pv_2\}} \Big(  (\neg \apred^{\pv_1} \land \neg 
\apred^{\pv_2}) \, \U \big( (\apred^{\pv_1} \yw{\land \F_{[0, \delta]}} \apred^{\pv_2}) \lor 
(\apred^{\pv_2} \yw{\land \F_{[0, \delta]}} \apred^{\pv_1}) \big) \Big) \geq 1 - 
\varepsilon.
\end{equation}
Note that~\eqref{eq:spec_tp} involves probability quantification over a set of 
paths, which cannot be expressed in non-hyper temporal logics.

\subsection{Probabilistic Anomaly Detectability}

An important feature of CPS is detectability of system anomalies, independently of the type of used \emph{sound} detector; this can be captured as probabilistic overshoot observability
on system outputs, where the input overshoot captures that an anomaly has occurred.
Specifically, we require that with probability of at least $1 - \varepsilon$ it holds
that: if (i)~in one execution, a signal $\pi$ steps (i.e., anomaly starts) and 
then stays bounded (e.g., within the modeled noise bound) for some time interval~$I$; and (ii)~in another execution, signal $\pi'$ steps and then overshoots (i.e., beyond the noise bound); then (iii)~the distance between the two signals is greater than a predefined threshold (i.e., the anomaly overshoot can be observed by a detector on system output). This is captured as the following~\hpstl formula
\begin{equation} \label{eq:spec_oo}
	\P^{\{\pi,\pi'\}} \bigg( \Big(\G \big(\mathit{step}^\pi  
	\Rightarrow \G_I (x^\pi < c)\big)  \wedge  
	\F\big(\mathit{step}^{\pi'}
	\wedge \F_I (x^{\yw{\pi'}} > c)\big) \Big)  \Rightarrow \
	\big(\F_I d(y^\pi, y^{\pi'}) > c' \big) \bigg) > 1 - \varepsilon,
\end{equation}
where $x$ is the input and $y$ is the output. 
As~\eqref{eq:spec_tp}, the formula in~\eqref{eq:spec_oo} also involves probability 
quantification over a~set~of~paths.

\subsection{Workload Fairness in Queueing Networks}

As shown in \cref{fig:queueing_network},
consider an embedded processing system with $n$ front servers and $m$ back servers (as in e.g.,~\cite{tang2015hardware}).
The requests (e.g., task, packets) arrive at each front-end queue, probabilistically over time, into buffers of different sizes.
For each queue, the requests are preprocessed
with probabilistic execution times,
and delivered to back servers with different buffer sizes, following some scheduling policy. We can probabilistically model the arrival and processing of requests by Markov Modulated Poisson Processes (MMPP)
of different parameters across all the 
servers~\cite{bolch2006queueing}. In the general case, 
this setup yields no easy exhaustive~solution.

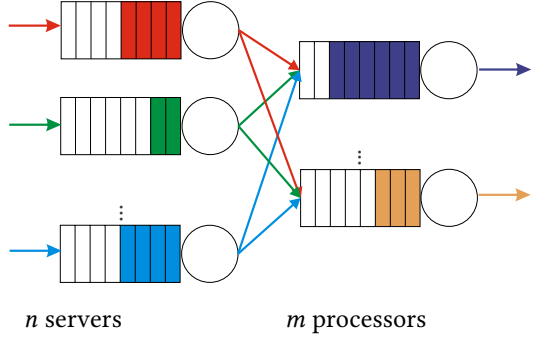
\begin{figure}[!t]
\centering
\def\svgwidth{0.5\linewidth}
\begingroup \makeatletter \providecommand\color[2][]{\errmessage{(Inkscape) Color is used for the text in Inkscape, but the package 'color.sty' is not loaded}\renewcommand\color[2][]{}}\providecommand\transparent[1]{\errmessage{(Inkscape) Transparency is used (non-zero) for the text in Inkscape, but the package 'transparent.sty' is not loaded}\renewcommand\transparent[1]{}}\providecommand\rotatebox[2]{#2}\newcommand*\fsize{\dimexpr\f@size pt\relax}\newcommand*\lineheight[1]{\fontsize{\fsize}{#1\fsize}\selectfont}\ifx\svgwidth\undefined \setlength{\unitlength}{238.55112535bp}\ifx\svgscale\undefined \relax \else \setlength{\unitlength}{\unitlength * \real{\svgscale}}\fi \else \setlength{\unitlength}{\svgwidth}\fi \global\let\svgwidth\undefined \global\let\svgscale\undefined \makeatother \begin{picture}(1,0.63048974)\lineheight{1}\setlength\tabcolsep{0pt}\put(0,0){\includegraphics[width=\unitlength,page=1]{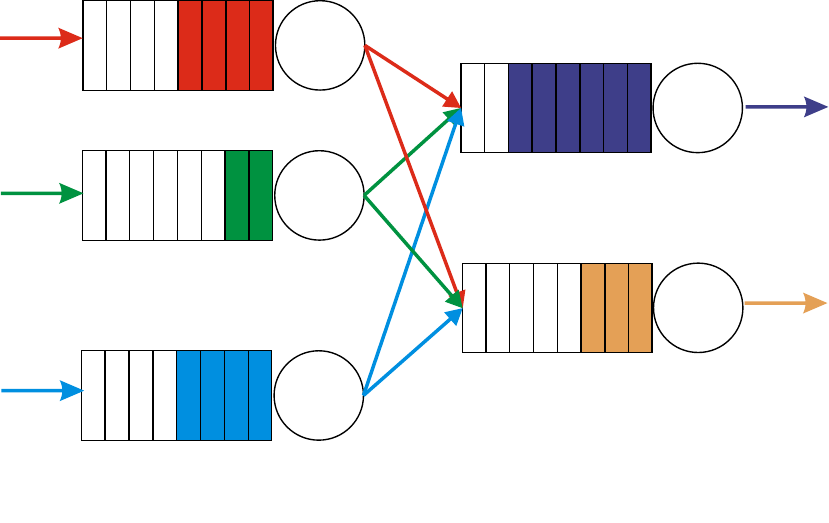}}\put(0.67398988,0.31853246){\color[rgb]{0.21568627,0.20392157,0.20784314}\rotatebox{90}{\makebox(0,0)[lt]{\lineheight{1.25}\smash{\begin{tabular}[t]{l}...\end{tabular}}}}}\put(0.21748546,0.21350742){\color[rgb]{0.21568627,0.20392157,0.20784314}\rotatebox{90}{\makebox(0,0)[lt]{\lineheight{1.25}\smash{\begin{tabular}[t]{l}...\end{tabular}}}}}\put(0.53097283,0.01434441){\color[rgb]{0,0,0}\makebox(0,0)[lt]{\lineheight{1.25}\smash{\begin{tabular}[t]{l}$m$ processors\end{tabular}}}}\put(0.03100454,0.01434441){\color[rgb]{0,0,0}\makebox(0,0)[lt]{\lineheight{1.25}\smash{\begin{tabular}[t]{l}$n$ servers\end{tabular}}}}\end{picture}\endgroup  \caption{Queueing Network.\label{fig:queueing_network}}
\end{figure}

Our goal is to check if a request-delivering policy between the front 
and back servers is fair~\cite{eryilmaz2006joint,georgiadis2006resource}; 
i.e., if a back server $i$ is more likely to be overloaded than another back 
server $j$.
Let $\tau_i$ and $\tau_j$ be the overload times of the back server $i$ and $j$, respectively.
We define a fairness property that with probability of at least $1 - 
\varepsilon$, given the overload time for the back server $i$, the back server 
$j$ is overloaded much earlier or later (more than some $t > 0$ than that 
time), with approximately equal probability (i.e., difference less than 
some~$\delta$)
\[
	\pr_{\pv_1} \Big( \big\vert \pr_{\pv_2} (\tau_i^{\pv_1} - \tau_j^{\pv_2} > t) - \pr_{\pv_2} (\tau_j^{\pv_2} - \tau_i^{\pv_1} > t) \big\vert < \delta \Big) > 1 - \varepsilon.
\]

To express this property in \hpstl, for $i \in [1, m]$, let $\apred_i$ be the predicate 
of the overload of the buffer of the $i^{th}$ back server.
The event that the back server $i$ is overloaded earlier than the back server 
$j$ more than time $\tau > 0$, can be expressed as $(\neg \apred_i^{\pv_1} 
\land \neg \apred_j^{\pv_2} ) \U (\apred_i^{\pv_1} \yw{\land \F_{[\tau,\infty)}} 
\apred_j^{\pv_2})$.
Hence, a fairness policy requirement  is captured by a \hpstl~formula
\begin{equation} \label{eq:spec_qn}
\begin{split}
	& \P^{\pv_1} \Big( \big\vert \P^{\pv_2} \big( (\neg \apred_i^{\pv_1} \land \neg \apred_j^{\pv_2} ) \U (\apred_i^{\pv_1} \yw{\land \F_{[\tau,\infty)}} \apred_j^{\pv_2}) \big)
	\\ & \hspace{6pt} - \P^{\pv_2} \big( (\neg \apred_i^{\pv_1} \land \neg \apred_j^{\pv_2} ) \U (\apred_j^{\pv_2} \yw{\land \F_{[\tau,\infty)}} \apred_i^{\pv_1}) \big) \big\vert \leq \delta \Big) \geq 1\hspace{-2pt} - \hspace{-2pt}\varepsilon.
\end{split}
\end{equation}
Note that~\eqref{eq:spec_qn} involves both comparison of probabilities and nesting of probability operators quantifying different paths, which are not allowable in common non-hyper temporal logics.

\section{Statistical Verification of \hpstl Properties}
\label{sec:smc}

\yw{In this section, we study the SMC of \hpstl on a PUS with given 
inputs.
As with previous works, we focus on handling the probability operators in HyperPSTL by sampling, which is the main issue for the SMC of probabilistic temporal logic. 
Through Sections~\ref{sub:pvalue} -- \ref{sub:nested}, we propose SMC 
algorithms based on Clopper-Pearson (CP) significance level calculation for all 
the ways the probability operators can be used or nested in \hpstl.
Accordingly, any nested HyperPSTL can be verified recursively by applying these 
SMC~algorithms.
The handling of temporal operators is similar to that of hyperSTL~\cite{nguyen_HyperpropertiesRealvaluedSignals_2017} and is thus not fully discussed due to the space limitations.
For a PUS, verifying bounded-time properties is straightforward;  verifying 
unbounded-time properties is more involving, and will be part of the future work.}

\subsection{SMC via CP Significance Level} \label{sub:pvalue}

A new feature of \hpstl compared to common temporal logics is the simultaneous 
probability quantification over multiple path variables. We now illustrate the 
idea of SMC for such formulas. Consider \hpstl formula $\Phi = (\P^\pvs \varphi 
< p)$, where (i) $\pvs = \{\pv_1, \ldots, \pv_K\} = \fv(\varphi)$ is the set of 
free path variables of $\varphi$, (ii) $p \in [0,1]$ is a probability 
threshold, and (iii) $\varphi$ contains no probability operator.
The semantics of $\Phi$ is
\[
\quad p_\varphi = \pr_{\Sig \sim \pathn{\abs{\pvs}} (\sysState)} \Big( \yw{\big( \sys}, V_\sysState[\pvs \rightarrow \Sig] \big) \models \varphi \Big) < p.
\]
\yw{The truth value of $\varphi$ can be evaluated on a set of concrete sample paths $\Sig = \{\sig_1, \ldots, \sig_K\}$ by assigning the concrete sample path $\sig_i$ to the free path variable $\sig_i$ for $i \in [K]$.} Hence, with a slight abuse of notation, we denote
\begin{equation} \label{eq:varphi}
	\varphi(\Sig) = \begin{cases}
		1, &\text{ if $\varphi$ is true on } \Sig, \\
		0, &\text{ otherwise.}
	\end{cases}
\end{equation}

\yw{Previous SMC approaches have used the sequential probability ratio test 
(SPRT) to evaluate $\Phi$ with the specification of an indifference 
margin~\cite{legay_StatisticalModelChecking_2015}.}
Specifically, assuming that
\begin{equation} \label{eq:assm_simple}
 	\abs{p_\varphi  - p} > \delta
\end{equation}
for some $\delta > 0$,
to evaluate $\Phi$, it suffices to test the two most 
indistinguishable cases, i.e., a Simple Hypothesis Testing (SHT) problem with 
two hypothesis,
\begin{equation} \label{eq:simp_check}
	\begin{split}
		H_0: p_\varphi = p - \delta, \quad H_1: p_\varphi = p + \delta,
	\end{split}
\end{equation}
which can then be solved by SPRT~\cite{hogg2005introduction}.

\yw{Since the choice of indifference margin is somewhat arbitrary,} we propose an indifferent margin-free SMC approach via significance level calculation.
For $i \in [N]$ and $K = \abs{\pvs}$,
let $(\sig^{(i)}_{1}, ..., \sig^{(i)}_{K})$
be a tuple of i.i.d. sample paths
drawn from the PUS $\sys$ starting from the state $\sysState$. Checking the correctness of $\varphi$ by~\eqref{eq:varphi}
on each tuple gives the sum statistic
\[
T = \sum\nolimits_{i \in [N]} \varphi \big( \sig^{(i)}_{1}, \ldots, \sig^{(i)}_{K} \big)
\]
that obeys the binomial distribution $\mathrm{Binom} (n, p_\varphi)$.
The average statistics $T/N$ is a unbiased estimator for $p_\varphi$.
Intuitively, when  $T/N < p$, it is more likely that $p_\varphi < p$; and the same for the other case.
Hence, we define the following statistical asserting function based on the samples
\begin{equation} \label{eq:assert_simple}
	\assert((\sys, \sysState) \models \Phi) = \begin{cases}
		1, &\text{ if } T/N < p \\
		0, &\text{ otherwise.}
	\end{cases}
\end{equation}

\yw{To ensure the asymptomatic correctness of the SMC algorithm, we assume that
\begin{equation} \label{eq:indifference_simple}
	p_\varphi \neq p,
\end{equation}
which is a weaker assumption than~\eqref{eq:assm_simple}.
When~\eqref{eq:indifference_simple} holds, as the number of samples increases, the samples will be increasingly concentrated on one side of $p$ by the central limit theorem. 
Therefore, a statistical analysis based on the majority of the samples has an increasing accuracy.
When~\eqref{eq:indifference_simple} is violated, the samples would be evenly distributed on the two sides of $p$, regardless of the sample size.
Thus, no matter how the sample size increases, the accuracy of any statistical test would not increase.
This will be illustrated later in the proof of \cref{thm:simple}.}

In general, the significance level for claiming $p_\varphi \in [a,b] \subseteq 
[0,1]$ \yw{(i.e., an upper bound of the probability of making a wrong claim)}, 
when $T/N \in [a,b]$, can be computed using a method from Clopper and 
Pearson~\cite{clopper1934use} by
\begin{equation} \label{eq:alpha_cp}
\alpha_{\text{CP}} (a, b \, \vert \, T, N) = 1 - 
	\begin{cases}
	(1-a)^{N} - (1-b)^{N}
	& \text{ if } T = 0 \\
	b^{N} - a^{N}
	& \text{ if } T = N \\
	F_\text{Beta} ( b \, \vert \, T+1, N - T ) - F_\text{Beta} ( a \, \vert 
\, T, N - T + 1 )
	& \text{ otherwise.}
\end{cases}
\end{equation}
where $F_\text{Beta}(\cdot \, \vert \, T_1, T_2)$ is the cumulative probability function (CDF) of the beta distribution $\text{Beta} (T_1, T_2)$ with the shape parameters $(T_1, T_2)$.
For computing the significance level of the assertion  $\assert((\sys, \sysState) \models \Phi)$, we utilize
\begin{equation} \label{eq:alpha_cp2}
	[a,b] = \begin{cases}
		[0,p], &\text{ if } T/N < p ,\\
		[p,1], &\text{ if }  T/N > p.
	\end{cases}
\end{equation}

\yw{\begin{remark} \label{rem:cp}
The CP significance levels~\eqref{eq:alpha_cp} have the following properties. 
First, $\alpha_{\text{CP}} (a, b \, \vert \, T, N)$ increases as the interval $[a,b]$ shrinks.
That is, for any $[a',b'] \subseteq [a,b] \subseteq [0,1]$, we have $\alpha_{\text{CP}} (a', b' \, \vert \, T, N) \leq \alpha_{\text{CP}} (a, b \, \vert \, T, N)$.
In addition, for $T \notin \{0, N\}$ and $N \gg 1$, we have that
\begin{equation} \label{eq:cp}
	\alpha_{\text{CP}} (a, b \, \vert \, T, N) \approx 1 - F_\text{Beta} ( b \, \vert \, T, N - T ) + F_\text{Beta} ( a \, \vert \, T, N - T ).
\end{equation}
Since the beta distribution $\text{Beta} (T, N - T)$ has the mean $T/N$ and the variance $T(N-T)/N^2 (N+1)$, 
for fixed $T/N$, as the number of samples $N \rightarrow \infty$,
the beta distribution becomes increasingly concentrated at $T/N$, and thus $\alpha_{\text{CP}} (a, b \, \vert \, T, N) \rightarrow 0$.
This implies that the probability of making the wrong claim decreases as more samples are available.
\end{remark}}

Given a desired significance level $\alpha$, we can design a new SMC algorithm by, at each iteration, collecting $B$ new samples, computing the CP significance interval and stopping when the result is less than $\alpha$, as summarized in \cref{alg:simple}.
Correctness of \cref{alg:simple} follows directly from the definition of significance~level. 

\begin{algorithm}[!t]
\caption{SMC of $(\sys, \sysState) \models \P^\pvs \varphi < p$.\label{alg:simple}}
\begin{algorithmic}[1]
\Require PUS $\sys$, desired significance level $\alpha_d$, batch size $B$.

\State $N \gets 0$, $T \gets 0$, $K \gets \yw{\abs{\pvs}}$, initial significance level $\alpha_\mathrm{CP} \gets 1$

\While{$\alpha_\mathrm{CP} > \alpha_d$}

\For{$i \in [n]$}

\State Draw $\sig_{N+1}, \ldots \sig_{N+B}$ from $\sysState$ in $\sys$.

\State $T \gets T + \sum_{i = N+1}^{N+B} \varphi (\sig^{(i)}_{1}, \ldots, \sig^{(i)}_{K})$; $N \gets N + B$.

\EndFor

\State Update $\assert$ by~\eqref{eq:assert_simple} and $\alpha_\mathrm{CP}$ by~\eqref{eq:alpha_cp} and~\eqref{eq:alpha_cp2}.

\EndWhile

\State \Return $\assert$  and $\alpha_\mathrm{CP}$.

\end{algorithmic}
\end{algorithm}

\begin{theorem} \label{thm:simple}
\cref{alg:simple} terminates with probability $1$ and gives the correct statistical assertion with probability at least $1-\alpha_d$.
\end{theorem}

\begin{proof}
\yw{{\bf Termination}:
For $p_\varphi \in \{0, 1\}$, the proof is trivial.
For $p_\varphi \in (0,1)$, from assumption~\eqref{eq:indifference_simple} and without loss of generality, let $p_\varphi < p$ and $\delta = p - p_\varphi$.
Recalling the second half of \cref{rem:cp}, 
for any $T/N \in [p_\varphi - \delta/2, p_\varphi + \delta/2] \subseteq (0, p)$, 
the variance of $\text{Beta} (T, N - T)$ is lower bounded by $\min_{x \in [p_\varphi - \delta/2, p_\varphi + \delta/2]} x (1 - x) / (N+1)$. 
Therefore, as $N \rightarrow \infty$, it uniformly converges to $0$. 
This implies that $\alpha_{\textrm{CP}} (0, p \, \vert \, T, N)$ uniformly converges to $0$ -- i.e., for any given $\alpha_d > 0$, there exists $N_0(p_\varphi, \delta) \in \nat$, such that $\alpha_{\text{CP}} (0, p \, \vert \, T, N) < \alpha_d$ for any $N \geq N_0(p_\varphi, \delta)$ and $T/N \in [p_\varphi - \delta/2, p_\varphi + \delta/2] \subseteq (0, p)$.

With $N \geq N_0(p_\varphi, \delta)$, by the law of large numbers (or central limit theorem), we have $\pr \big( T/N \in [p_\varphi - \delta/2, p_\varphi + \delta/2] \big) \rightarrow 1$, as the number of samples $N \rightarrow \infty$. 
Therefore, \cref{alg:simple} terminates with probability $1$.}

{\bf Correctness}:
Let $\tau$ be the step \cref{alg:simple} terminates
and $A$ be ``the assertion $\assert$ in~\eqref{eq:assert_simple} is correct'',
then
$\pr(A) = \sum\nolimits_{i \in \nat} \pr(A   \vert \, \tau = i) \pr(\tau = i)$.
By construction of the significance intervals, for any $i \in \nat$, we have $\pr(A \allowbreak \, \vert \, \tau = i) > 1 - \alpha_d$.
In addition by {\bf Termination}, we have $\sum_{i \in \nat} \pr(\tau = i) = 1$, 
Thus, $\Pr(A) \geq 1 - \alpha_d$.
\end{proof}

\subsection{SMC of Joint Probabilities} \label{sub:joint}

Another new feature of \hpstl is the arithmetics and comparisons of the probabilities of multiple sub-properties.
\yw{For example, we can compare the satisfaction probability of $\varphi_1$ and $\varphi_2$ by the \hpstl formula $p_1 < p_2$, where $p_1 = \P^{\pvs_1} \varphi_1$ and $p_2 = \P^{\pvs_2} \varphi_2$, according to the syntax~\eqref{eq:hpstls_1}, \eqref{eq:hpstls_2}.
For simplicity, let $\pvs_1 = \fv(\varphi_1)$ and $\pvs_2 = \fv(\varphi_2)$. 
Mathematically, this is equivalent to reasoning over the joint probabilities of these properties.
Specifically, $p_1 < p_2$ can be equivalently expressed as a specification on the joint probability $(p_1, p_2) \in D$, where $D = \set{(x_1, x_2) \in [0,1]^2}{x_1 < x_2}$.

To formally capture this, we introduce an additional syntactic rule $(p_1, \ldots, p_n) \in D$ in~\eqref{eq:hpstls_1}, whose semantics is given by 
\yw{\begin{equation} \label{eq:new_rule}
(\sys, \sysState) \models (p_1, \ldots, p_n) \in D \Leftrightarrow
\sys \models \big( \istx{p_1}{\sysState}, \ldots, \istx{p_n}{\sysState} \big) \in D,
\end{equation}}
where $D \subseteq [0,1]^n$ is measurable.}

While the expressiveness of \hpstl is unchanged with the new rule~\eqref{eq:new_rule}, the conjunction and disjunction of several \hpstl formula can be simplified.
For example, the \hpstl formula $\Phi_1 \land \Phi_2$~with
\begin{equation}
\begin{split}
	& \Phi_1 = \big( f_1(\P^{\pvs_1} \varphi_1, f_2 (\P^{\pvs_2} \varphi_2, 
\P^{\pvs_3} \varphi_3)) > c_1 \big)
	\\ & \Phi_2 = \big( f_3(\P^{\pvs_2} \varphi_2) ) < c_2 \big)
\end{split}
\end{equation}
can be equivalently written as $(\P^{\pvs_1} \varphi_1, \P^{\pvs_2} \varphi_2, \P^{\pvs_3} \varphi_3) \in D$, where
\[
	D = \big\{(x_1, x_2, x_3) \in [0,1]^3 \, \vert \, f_1(x_1, f_2 (x_2, x_3)) > c_1, f_3(x_2) < c_2 \big\}.
\]

In addition, the new rule simplifies the SMC of the specification.
Previously, it requires checking both $\Phi_1$ and $\Phi_2$ separately, and then a probabilistic composition of the two results.
With the new rule, a single procedure of checking whether the joint probability $(\P^{\pvs_1} \varphi_1, \P^{\pvs_2} \varphi_2, \P^{\pvs_3} \varphi_3)$ is in $D$ is~sufficient.

We now demonstrate the idea of verifying the joint probability by checking a non-nested state formula
\[
(\sys, \sysState) \models \Ps \in D,
\]
where for $i \in [n]$, and $\varphi_i$ contains no probability operator and $\pvs_i = \fv(\varphi_i)$.
As a statistical approach is adopted, we assume that the exact probability of 
satisfying $\varphi$ does not lie within the boundary of the test region $D$, 
as stated in~\cref{ass:indifference}.

\begin{assumption} \label{ass:indifference}
To check $(\sys, \sysState) \models \Ps \in D$, we assume that
(i) the test region $D$ is a simply connected domain with $\bm (D) \neq 0$, and (ii)
\[
\Big( \pr_{\Sig_1 \sim \pathn{\abs{\pvs_1}} (\sysState)} \big( \yw{\big( \sys}, V_\sysState[\pvs_1 \rightarrow \Sig_1] \big) \models \varphi_1 \big), \ldots, \pr_{\Sig_n \sim \pathn{\abs{\pvs_n}} (\sysState)} \big( \yw{\big( \sys}, V_\sysState[\pvs_n \rightarrow \Sig_n] \models \varphi_n \big) \Big) \notin \partial D.
\]
\end{assumption}

\yw{\cref{ass:indifference} can be viewed as the multidimensional generalization of~\eqref{eq:indifference_simple}, and is necessary for asymptomatic correctness of the SMC algorithm, as discussed in \cref{sub:pvalue}.}

\begin{remark} \label{rem:indifference_margin}
Compared to previous studies on SMC using sequential probability ratio tests (SPRT)~\cite{zuliani_StatisticalModelChecking_2015,wang_StatisticalVerificationPCTL_2018}, \cref{ass:indifference} is weaker as it requires no a priori knowledge on the indifference margin.
\end{remark}

\begin{remark}
From~\cref{ass:indifference}, we have that $\Ps \allowbreak \in D_1$ and $\Ps \in D_2$  are semantically equivalent if $\overline{D_1} = \overline{D_2}$.
In addition, $\P^{\pvs} \in D$ and $\P^{\pvs} (\neg \varphi) \in D^c$ are semantically equivalent when $D^c = [0,1] \backslash D$.
\end{remark}

By the semantic rule~\eqref{eq:new_rule}, the SMC problem is converted to a 
composite hypothesis testing problem
\begin{equation} \label{eq:ht_multiple}
\begin{split}
& \begin{cases}
	H_0: & (p_{\varphi_1}, \ldots p_{\varphi_n}) \in D, \\
	H_1: & (p_{\varphi_1}, \ldots p_{\varphi_n}) \in [0,1]^n \backslash D,
\end{cases}
\\ & \quad p_{\varphi_i} = \pr_{\Sig_i \sim \pathn{\abs{\pvs_i}} (\sysState)} \big( \yw{\big( \sys},V_\sysState[\pvs_i \rightarrow \Sig_i] \big) \models \varphi_i \big)
\text{ for } i \in [n].
\end{split}
\end{equation}
For each $i \in [n]$, let $\{\Sig_j\}_{j \in [N_i]}$ be $N_i$ tuples of i.i.d sample paths of the PUS $\sys$ starting from the same state $\sysState$ that are used to estimate $p_{\varphi_i}$.
Similar to \cref{sub:pvalue} approach, we consider the statistics
\begin{equation}
	T_i = \sum\nolimits_{j = 1}^{N_i} \varphi_i (\Sig^{(j)}_{1}, \ldots \Sig^{(j)}_{K_i}), \quad K_i = \abs{\pvs_i}.
\end{equation}
where $T_i \sim \Binom(N_i , p_{\varphi_i})$.
We define the assertion using the average statistics $T_i/ N_i$ for $i \in [n]$ as
\begin{equation} \label{eq:assert_join}
	\assert((\sys, \sysState) \models \Phi) = \begin{cases}
		1 &\text{ if } \Big( \frac{T_1}{N_1}, \ldots, \frac{T_n}{N_n} 
\Big) \in D, \\
		0 &\text{ otherwise.}
	\end{cases}
\end{equation}
For this multi-dimensional case, to compute the exact CP significance level for a general domain $D$ involves multi-dimensional integrations on it, which can be \yw{computationally intensive}.
Therefore, we compute an upper bound $\alpha_\mathrm{CP}$ on the significance level of the assertion~\eqref{eq:assert_join} by finding a hypercube $\prod_{i \in [n]} [a_i, b_i] $ such that
\begin{equation} \label{eq:multiple_mle}
	\Big( \frac{T_1}{N_1}, \ldots, \frac{T_n}{N_n} \Big) \in \prod\nolimits_{i \in [n]} [a_i, b_i] \subset D.
\end{equation}
Due to the monotonicity of significance levels, the significance level of $( T_1/N_1, ..., T_n/N_n ) \in D$ is upper bounded by that of  $( T_1/N_1, ..., \allowbreak T_n/N_n ) \in \prod_{i \in [n]} [a_i, b_i]$,
which can be computed directly by compositing the significance level of $T_i/ N_i \in [a_i, b_i]$ for $i \in [n]$, using the results in \cref{sub:pvalue}.
Thus, we compute an upper bound $\yw{\bar{\alpha}}_\mathrm{CP}$ of the exact CP significance level as
\begin{equation} \label{eq:confidence_multiple}
	\yw{\bar{\alpha}}_\mathrm{CP} = 1- \prod\nolimits_{i=1}^n \yw{\alpha}_{\mathrm{CP}} (a_i, b_i \, \vert \, T_i, N_i),
\end{equation}
where $\yw{\alpha}_{\mathrm{CP}} (a_i, b_i \, \vert \, T_i, N_i)$ is defined in~\eqref{eq:alpha_cp}.

\yw{When implementing this verification approach, for each iteration, we look for a hypercube $\prod_{i \in [n]} [a_i, b_i]$ satisfying~\eqref{eq:multiple_mle} with $a_i < b_i$ for $i \in [n]$. 
Although finding such a hypercube is only possible if $(T_1/N_1, ..., T_n/N_n) \notin \partial D$, this is guaranteed with probability $1$ for large samples when \cref{ass:indifference} holds.

To minimize the upper bound of the significance level $\bar{\alpha}_\mathrm{CP}$ in~\eqref{eq:confidence_multiple}, the hypercube should be preferably as large as possible.
For a simple domain $D$, the analytic solutions of such a largest hypercube can be derived directly as a function of $T_i, N_i$ for $i \in [n]$ and the functions defining $\partial D$
More generally, especially if $D$ is convex, the largest hypercube can be derived by solving the optimization problem of maximizing its volume while keeping it inside $D$.
Admittedly, solving the optimization problem at every iteration can still be inefficient for some cases.
To remedy for this, we can (1) reduce the frequency of computing the significance level by drawing samples in batches; and (2) search only for approximate maxima in~optimization.

Finally, we note that the upper bound $\bar{\alpha}_\mathrm{CP}$ is asymptotically tight if a largest hypercube is used to compute it.
This holds because, by the law of large numbers, as the number of samples increases, $( T_1/N_1, ..., T_n/N_n )$ concentrates near $(p_{\varphi_1}, \ldots, p_{\varphi_n})$ and the largest hypercube converges to a constant one strictly containing $(p_{\varphi_1}, \ldots, \allowbreak p_{\varphi_n})$.
Thus, the probability of $( T_1/N_1, ..., \allowbreak T_n/N_n ) \in \prod_{i \in [n]} [a_i, b_i]$ converges to that of that of $( T_1/N_1, ..., \allowbreak T_n/N_n ) \in D$.}

Based on the previous discussions, we derive~\cref{alg:joint}.
Correctness of \cref{alg:joint} is given by \cref{thm:joint} that can be proved in the same way as \cref{thm:simple}.

\begin{theorem} \label{thm:joint}
\cref{alg:joint} terminates with probability $1$ and gives the correct statistical assertion with probability at least $1-\alpha$.
\end{theorem}

\begin{algorithm}[!t]
\caption{SMC of $(\sys, \sysState) \models \Ps \in D$.}
\label{alg:joint}
\begin{algorithmic}[1]
\Require PUS $\sys$, desired significance level $\alpha_d$, batch size $B$.

\State $N_1, \ldots, N_n \gets 0$, $\yw{\bar{\alpha}}_\mathrm{CP} \gets 1$

\For{$i \in [n]$}
\State $K_i \gets \abs{\pvs_i}$, $T_i \gets 0$
\EndFor

\While{$\yw{\bar{\alpha}}_\mathrm{CP} < 1- \alpha_d$}

\For{$i \in [n]$}

\State Draw $\sig_{N_i+1}, \ldots \sig_{N_i+B}$ from $\sysState$ in $\sys$.

\State $T_i \gets T_i + \sum_{j = N_i + 1}^{N_i + B} \varphi_i (\sig^{(j)}_{1}, \ldots \sig^{(j)}_{K_i})$; $N_i \gets N_i + B$.

\EndFor

\State Update $\assert$ by~\eqref{eq:assert_join} and $\yw{\bar{\alpha}}_\mathrm{CP}$ by~\eqref{eq:confidence_multiple}.

\EndWhile

\State \Return $\assert$  and $\yw{\bar{\alpha}}_\mathrm{CP}$.

\end{algorithmic}
\end{algorithm}

\subsection{SMC of Nested Probability Operators} \label{sub:nested}

Finally, we consider SMC of nested \hpstl formulas on the PUS $\sys$. By the syntax of \hpstl in \cref{sub:syntax}, a nested \hpstl formula is constructed iteratively in two ways: (i) replacing an atomic proposition with a general state formula, and (ii) consecutively nested quantification of free path variables in a non-nested formula.
The former also appears in common temporal logics and the latter is unique to \hpstl.

For (i), we show the idea of SMC by checking the satisfaction on a state $\sysState$ of the nested formula
$\Psi = \P^{\pvs} \psi [\rho] \in D_1$ derived by replacing an atomic proposition of $\psi$ with a non-nested state formula $\rho = \Ps \in D_2$.
In $\Psi$, we have $D_1 \subseteq [0,1]$, $D_2 \subseteq [0,1]^n$.
If $\rho$ is treated as an atomic proposition, $\Psi$ becomes a non-nested state \hpstl formula as discussed in \cref{sub:pvalue}.
The \hpstl state formulas nesting more probability operators in this fashion can be statistically verified in the same way.

To verify $\Psi$, we follow a compositional analysis similar to~\cite{sen_StatisticalModelChecking_2005a,sen_VESTAStatisticalModelchecker_2005a}.
If the state space $\sysStates$ of the PUS is finite (e.g., a CTMC from \cref{sub:ctmc}), we can statistically verify the sub-formula $\rho$ on each state $\sysState$ of the PUS $\sys$ with significance level $\alpha_\sysState$, using \cref{thm:joint} and \cref{alg:joint}, and label the state with $\rho$ if the assertion given by \cref{alg:joint} is $\assert \big( V \models \rho \big) = 1$.
Then, we can statistically verify the full formula $\Psi$ on the relabeled PUS as a non-nested formula with significance level $\alpha_0$ using \cref{thm:joint}.
The overall significance level is $\alpha = \sum_{\sysState \in \sysStates} \alpha_\sysState + \alpha_0$, which is only bounded when $\sysStates$ is finite.
The SMC for $\Psi$ on infinite-state PUS provides an avenue for future work.
For this work, it brings no limitation as we only verify such nested formulas on finite-state PUS (e.g., queueing in \cref{sec:evaluation}), while for \phioa~such formulas are not used to capture properties~of~interest.

To implement the SMC algorithm for $\Psi$ on finite-state PUS, given the overall significance level $\alpha$, we can split it into the summation $\sum_{\sysState \in \sysStates} \alpha_\sysState + \alpha_0$.
The simplest way is $\alpha_\sysState = \alpha_0 = \alpha/(\abs{\sysStates} + 1)$ for $\sysState \in \sysStates$, where $\sysStates$ is the number of states of the PUS $\sys$.
Then, we can employ \cref{alg:joint} to verify $\rho$ on each state $\sysState$ with significance level $\alpha_\sysState$ and then assert $\Psi$ with significance level $\alpha_0$ using \cref{alg:joint} again.
This is summarized by \cref{alg:nested1}.

\begin{algorithm}[!t]
\caption{SMC of $(\sys, \sysState) \models \P^{\pvs} \psi [\rho] \in D_1$ with $\rho = \Ps \in D_2$.}
\label{alg:nested1}
\begin{algorithmic}[1]
\Require PUS $\sys$, desired significance level $\alpha_d$.

\State Split $\alpha_d$ into $\sum_{\sysState \in \sysStates} \alpha_\sysState + \alpha_0$

\For{$\sysState \in \sysStates$}

\State Verify $(\sys, \sysState) \models \rho$ on $\sys$ by \cref{alg:joint} with significance level $\alpha_\sysState$ and label $\sysState$ with $\rho$ if the assertion is positive.

\EndFor

\State Verify $(\sys, \sysState) \models \P^{\pvs} \psi [\rho] \in D_1 $ on the relabeled $\sys$ by \cref{alg:joint} with significance level $\alpha_0$.
\end{algorithmic}
\end{algorithm}

For (ii), we show the idea of SMC by checking the satisfaction on a state $\sysState$ of a \hpstl formula $\Psi = \P^{\pvs_1} (\P^{\pvs_2} \varphi < p_2) < p_1$, where the sub-formula $\varphi$ contains no probability operator and all its path variables are probabilistically quantified by $\P^{\pvs_1} \P^{\pvs_2}$,
i.e., $p_1, p_2 \in [0,1]$, $\pvs_1 \cap \pvs_2 = \emptyset$, and $\fv(\varphi) = \pvs_1 \cup \pvs_2$.

The SMC can be similarly done for \hpstl state formulas that nest more probability operators in this fashion.
The formula $\Psi$ says that with probability at most $p_1$, we can find a set of paths $\pvs_1$ such that the probability to find another set of paths $\pvs_2$ to satisfy $\varphi$ is at most $p_2$.
This formula can be equivalently expressed by $\Psi = \P^{\pvs_1}  \big( \P^{\pvs_2} \in [0, p_2] \big) \in [0, p_1]$ using the  rule~\eqref{eq:new_rule}.
We note that for (ii),  unlike (i), the state space $\sysStates$ of the PUS can be infinite.

The main idea is as follows.
For $i \in [N]$, let $\Sig_i$ be a $\abs{\pvs_1}$-tuple of i.i.d. sample paths starting from $\sysState$ in the PUS; define the indicator
\begin{equation} \label{eq:partial}
	T_i = \id \Big( (\sys, \sysState) \models \llbracket \P^{\pvs_2} \varphi < p_2 \rrbracket_{V [\pvs_1 \rightarrow \Sig_i]} \Big)
\end{equation}
of whether the partly instantiated formula $(\sys, \sysState) \models \llbracket \P^{\pvs_2} \varphi < p_2 \rrbracket_{V[\pvs_1 \rightarrow \Sig_i]}$ is true under this instantiating.
Although $T_i$ is not directly accessible, we can estimate it statistically using the assertion $A_i$ of \cref{alg:joint} for any given significance level $\alpha_1 > 0$.

To verify the full formula $\Psi$, we only need to estimate the total number of positive instantiation $T = \sum_{i \in [N]} T_i$.
We estimate it using $A = \sum_{i \in [N]} A_i$.
Since $A_i \neq T_i$ with probability at most $\alpha_1$ for all $i \in [N]$, we have $\abs{T-A} < \Delta$ with the significance level
\begin{equation} \label{eq:delta}
	\alpha_2 = 1 - F_{\mathrm{Binom}} (\Delta \, \vert \, N, \alpha_1),
\end{equation}
where $F_{\mathrm{Binom}}$ is the Binomial cumulative distribution function.
Since
\begin{equation} \label{eq:Ti}
	T \in [T_1, T_2] = [\min\{0, A - \Delta\}, \max\{A + \Delta, N\}],
\end{equation}
we can check the full formula $\Psi$ using these minimal and maximal estimations of $T$.
\yw{Intuitively, if $T_2/N < p$ (hence $T_1/N < p$), it is more likely that $\Psi$ is true; if $T_1/N > p$ (hence $T_2/N > p$), it is more likely that $\Psi$ is false; otherwise, further sampling is needed. 
Thus, we define the following statistical asserting function by}
\begin{equation} \label{eq:assert}
	\assert((\sys, \sysState) \models \Psi) = \begin{cases}
		1 &\text{ if } T_2/N  < p_1 \\
		0 &\text{ if } T_1/N  > p_1 \\
		\text{undecided}, &\text{ otherwise.}
	\end{cases}
\end{equation}
When a final assertion is possible, its significance level is the larger one between plugging $T_1$ and $T_2$ into~\eqref{eq:alpha_cp} and~\eqref{eq:alpha_cp2}.
Accordingly, the overall significance level $\alpha$ is
\begin{equation} \label{eq:alpha}
	\alpha = \alpha_2 +
	\begin{cases}
		\alpha_\mathrm{CP} (0, p_1 \, \vert \, T_2, N) &\text{ if } 
T_2/N  < p_1 \\
		\alpha_\mathrm{CP} (p_1, 1 \, \vert \, T_1, N) &\text{ if } 
T_1/N  > p_1
	\end{cases}
\end{equation}
where $\alpha_\mathrm{CP}$ is given by~\eqref{eq:alpha_cp}.

To implement the SMC algorithm, given the overall significance level $\alpha$, we need to simultaneously decrease both the significance level $\alpha_1$ for making assertions on the partly instantiation of $\psi$ for given values of $\pvs_1$, and the significance level $\alpha_2$ for estimating the sum of those assertions.
We start from $\Delta = c \alpha_1 N$ with $c = 1$ and $\alpha_1 = \alpha_d$.
When $\alpha_2$ is the main source of statistical error, i.e.,  $\alpha_2 > \alpha/2$, we decrease $\alpha_2$ by increasing the parameter $c$ by $1$;
otherwise, we decrease $\alpha_1$ by reducing it by half.
This is summarized in \cref{alg:nested2}.

\begin{algorithm}[!t]
\caption{SMC of $\P^{\pvs_1} (\P^{\pvs_2} \varphi < p_2) < p_1$.}
\label{alg:nested2}
\begin{algorithmic}[1]
\Require PUS $\sys$, desired significance level $\alpha_d$.

\State Set initial significance levels $\alpha_1$ for $i \in [N]$.

\State $T \gets 0$, $N \gets 0$, $T_i \gets 0$ for $i \in [N]$.

\State $c \gets 1$, $\alpha \gets 1$, $\alpha_1 \gets \alpha_d$.

\While{$\alpha > \alpha_d$}

\State $N \gets N+1$.
Draw $\Sig_{N+1}$ staring from $\sysState$ in $\sys$.

\For{$i \in [N]$}

\State Update $A_i$ with significance level $\alpha_1$ by \cref{alg:joint}.

\EndFor

\State $A \gets \sum_{i \in N} A_i$ , $\Delta \gets c \alpha_1 N$.

\State Update $\alpha_2$ by~\eqref{eq:delta}, $T_1, T_2$ by~\eqref{eq:Ti} and $\alpha$ by~\eqref{eq:alpha}.

\If{$\alpha_2 > \alpha/2$}
$c \gets c + 1$,

\Else
$~~\alpha_1 \gets \alpha_1/2$.

\EndIf

\EndWhile

\State \Return Assertion given by~\eqref{eq:assert}.

\end{algorithmic}
\end{algorithm}

\section{Evaluation} \label{sec:evaluation}

We numerically evaluate our SMC algorithms on several benchmarks with
different complexity levels.
Other probabilistic hyperproperties on different systems are handled in a
similar manner, but due to space constraints here we focus on the discussed 
properties/systems.
{All benchmarks are implemented in Matlab/Simulink and are available
in~\cite{simulink}.
Specifically, the Toyota powertrain model is derived 
from~\cite{jin_BenchmarksModelTransformations_2014}; and the queueing networks 
are implemented in Simulink using the SimEvents Toolbox~\cite{simevents}.}
Evaluations are performed on a laptop with 16~GB~RAM and Intel Xeon E-2176 CPU.

\yw{For each benchmark, we evaluate the proposed SMC algorithms in different setups by changing the desired significance level $\alpha$, as well as the parameters (e.g., $\delta$, $\varepsilon$, and $t$) in the objective \hpstl specifications.
The proposed SMC algorithms are executed repeately on each setup for $100$ times.
This is to check whether the probability for the proposed SMC algorithms to return the correct assertion is at least $1 - \alpha$: i.e., we repeat the SMC algorithm for each setup for $100$ times, and check if it makes the correct assertion for at least $100 (1 - \alpha_d)$ times.

For each setup, to compare it with the assertions of the proposed SMC algorithms, the truth value of the \hpstl specification of interest is derived by estimating the probabilities involved in it using an analytic solution or numerous sampling.
Specifically, for the thermostat, we derive an analytic solution for the left hand side of~\eqref{eq:spec_tp} from its dynamics; this can be  done due to its simplicity. 
For the powertrain, we estimate the left hand side of~\eqref{eq:spec_tp} by sampling $10^5$ pairs of $(\pv_1, \pv_2)$, for which the standard error is less than $0.01$.
For the queueing networks, we estimate the left hand side of~\eqref{eq:spec_qn} by drawing $500$ samples for $\pv_1$.
For each sample of $\pv_1$, we draw $500$ samples for $\pv_2$ to evaluate the truth value of $\big\vert \pr_{\pv_2} (\tau_i^{\pv_1} - \tau_j^{\pv_2} > t) - \pr_{\pv_2} (\tau_j^{\pv_2} - \tau_i^{\pv_1} > t) \big\vert < \delta$.
The total standard error for estimating the left hand side of~\eqref{eq:spec_qn} is less than $0.05$.}

Results for all considered setups are shown
in~\cref{tb:thermostat,tb:tp,tb:qn,tb:qn2}; as can be seen, the estimated
accuracy of our SMC algorithms are very close to~$1$, showing the
conservativeness of the CP significance level.
We also report the average number of sample paths and the execution times for each setup based on the $100$ repetitions, rounded up to their standard errors.
In~all~simulations, the number increases when the desired significance level decreases, showing the trade-off between accuracy and the sampling~cost.
\yw{In addition, the execution time is mainly consumed by drawing samples from the Simulink models, and is approximately propositional to the number of samples.}

\subsection{Thermostat} 
A thermostat can be modeled as a simple \phioa{} with
two modes $\mathsf{Heat}$ and $\mathsf{Cool}$ (\cref{fig:thermostat}),
and one state variable $T$
that varies \yw{within the temperature interval}
$[T_l, T_h] = [15, 40] \subseteq \real$.
The mean heating and cooling rates are
$c_1 = c_2 = 5$;
they are subject to time-invariant but random Gaussian error
$n_1, n_2 \sim \mathbf{N} (0, 0.5^2)$.
The thermostat starts from $\big( T = T_l, \mathsf{Heat} \big)$.
We verify the sensitivity of the running period of a heat and cool cycle
under the noise,
which is represented in \hpstl by~\eqref{eq:spec_tp}
with $\apred := \big( T = T_l, \mathsf{Cool} \big)$.
We statistically verified the sensitivity specification using~\cref{alg:joint}, with specification~\eqref{eq:spec_tp} parameters $\delta \in \{0.9, 1.1\}$ and $\varepsilon\in\{0.05, 0.01\}$, under the desired significance levels $\alpha \in \{0.01,0.05\}$.

The derived results in \cref{tb:thermostat} give accurate estimations (with significance level as low as $\alpha = 0.01$) on the probability distribution of the sensitivity, with a relatively small number of samples (at~most a few hundred samples for each setup).
We verified that the sensitivity of the running period of the thermostat is less than $\delta=1.1$ with probability $1 - \varepsilon = 0.95$, but not less than $\delta=0.9$ with the same probability, showing that the $\mathbf{0.95}$ percentile is between $\mathbf{[0.9, 1.1]}$.
Also, as the sensitivity specification is false for $1 - \varepsilon = 0.99$ for both $\delta = 0.9$ and $\delta = 1.1$, showing that the $\mathbf{0.99}$ percentile is in $\mathbf{[1.1, \infty]}$.

\begin{figure}
\centering
\begin{tikzpicture}
\node[draw, minimum width=1.3cm, align=center] (l) {Heat: $\dot{T} = c_1 + n_1$};
\node[draw, minimum width=1.3cm, right of=l, xshift=3cm, align=center] (u) {Cool: $\dot{T} = - (c_2 + n_2)$};
\draw[->] (l) to[bend left=5] node[above] {$T=T_u$} (u);
\draw[->] (u) to[bend left=5] node[below] {$T=T_l$} (l);
\end{tikzpicture}
\caption{Dynamical Model of a Thermostat.\label{fig:thermostat}}
\end{figure}
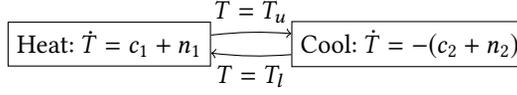

\begin{table}[!t]
\centering
\setlength\tabcolsep{2mm}
\begin{tabular}{rrrrrrr}
$\delta$ & $\varepsilon$ & $\alpha$ & Acc. & Sam. & Time (s) & Ans.  \\
\toprule
 0.9 &  0.05 &  0.05 &  1.00 &  1.8e+02 & \yw{2.1e+02} & \False \\
 0.9 &  0.05 &  0.01 &  1.00 &  5.0e+02 & \yw{6.0e+02} & \False \\
 0.9 &  0.01 &  0.05 &  1.00 &  2.8e+01 & \yw{3.4e+01} & \False \\
 0.9 &  0.01 &  0.01 &  1.00 &  4.6e+01 & \yw{5.5e+01} & \False \\
 1.1 &  0.05 &  0.05 &  0.99 &  3.0e+02 & \yw{3.5e+02} & \True \\
 1.1 &  0.05 &  0.01 &  0.99 &  6.1e+02 & \yw{7.3e+02} & \True  \\
 1.1 &  0.01 &  0.05 &  1.00 &  1.3e+02 & \yw{1.6e+02} & \False \\
 1.1 &  0.01 &  0.01 &  1.00 &  2.2e+02 & \yw{2.6e+02} & \False \\
\bottomrule
\end{tabular}
\caption{Accuracy (Acc.), average number of samples (Sam.), \yw{average execution time (Time)}, and SMC results (Ans.) for sensitivity~\eqref{eq:spec_tp} with parameters $\delta$ and $\varepsilon$ of Thermostat under significance level $\alpha$ (note the $1 - \alpha$ accuracy guarantee by our SMC method).\label{tb:thermostat}}\end{table}

\subsection{Toyota Powertrain Control System}
We use the Simulink~model for the Toyota Powertrain with a four-mode embedded controller from~\cite{jin_PowertrainControlVerification_2014}.
It can be considered as a \phioa~ \allowbreak with four modes
and $15$ state variables.
We consider the sensitivity of the recovery time after start
of the deviation percentage of the air/fuel (A/F) ratio $\mu$
to the level $\abs{\mu} < 0.05$,
under the mean RPM $2500$ subjecting to Gaussian noise $\mathbf{N} (0,25^2)$.
The sensitivity specification is formally expressed by \hpstl formula~\eqref{eq:spec_tp}
with $\apred = \big( \abs{\mu} < 0.05 \big)$; we
 statistically verified it using~\cref{alg:joint}
with parameters $\varepsilon\in\{0.05, 0.01\}$ and $\delta \in \{0.06,0.07\}$, under the desired significance levels $\alpha \in \{0.01,0.05\}$.

The results shown in \cref{tb:tp} give accurate estimations (with significance level as low as $0.01$) on the probability distribution of the sensitivity under the given embedded controller; this is achieved with a relatively small number of samples (at most a few hundred samples for each setup).
We verified that the sensitivity of the recovery time of the powertrain is less than $\delta=0.15s$ with probability $1 - \varepsilon = 0.99$, but not less than $\delta=0.20s$ with the same probability, showing that the $\mathbf{0.99}$ percentile is between $\mathbf{[0.15s, 0.20s]}$.
Also, as the sensitivity specification is true for $1 - \varepsilon = 0.95$ for both $\delta = 0.15s$ and $\delta = 0.20s$, showing that the  $\mathbf{0.95}$ percentile is in $\mathbf{[0, 0.15s]}$.

\begin{table}[!t]
\centering
\setlength\tabcolsep{2mm}
\begin{tabular}{rrrrrrr}
$\delta$ & $\varepsilon$ & $\alpha$ & Acc. & Sam. & Time (s) & Ans. \\
\toprule
 0.15 &  0.95 &  0.05 &  1.00 &  5.9e+01 & \yw{8.1e+00} & \True \\
 0.15 &  0.95 &  0.01 &  1.00 &    9.0e+01 & \yw{1.3e+01} & \True \\
 0.15 &  0.99 &  0.05 &  0.99 &  6.6e+01 & \yw{9.1e+00} & \False \\
 0.15 &  0.99 &  0.01 &  1.00 &  9.7e+01 & \yw{1.4e+01} & \False \\
 0.20 &  0.95 &  0.05 &  0.98 &  5.9e+01 & \yw{8.1e+00} & \True \\
 0.20 &  0.95 &  0.01 &  1.00 &    9.0e+01 & \yw{1.2e+01} & \True \\
 0.20 &  0.99 &  0.05 &  1.00 &    3.0e+02 & \yw{4.2e+01} & \True \\
 0.20 &  0.99 &  0.01 &  0.99 &  4.6e+02 & \yw{1.8e+02} & \True \\
\bottomrule
\end{tabular}
\caption{Accuracy (Acc.), average number of samples (Sam.), \yw{average execution time (Time)}, and SMC results (Ans.) for sensitivity~\eqref{eq:spec_tp} with parameters $\delta$ and $\varepsilon$ of Toyota Powertrain under significance level $\alpha$ (note the $1 - \alpha$ accuracy guarantee by our SMC method).\label{tb:tp}}\end{table}

\subsection{Queueing Networks} We consider two queuing networks \yw{with different sizes $m$ and $n$}, as shown in \cref{fig:queueing_network}:
the {\bf Small} has $1$ front-end servers and $2$ back-end servers,
the {\bf Large} has $25$ front-end servers and $20$ back-end servers.
For the small model, the package arrival and processing are modeled by exponential distribution, while in the second model, the package arrival and processing are modeled by Markov Modulated Poisson Processes 
with different parameters over all queues.
Note that the parameters for different front-end  servers are not identical,
so deriving an exhaustive solution is virtually impossible for large sizes.
We consider the fairness of workloads in queuing networks under the policy
that the front servers deliver to the back server with the shortest queue.
The fairness specification is formally defined in \hpstl by~\eqref{eq:spec_qn} with $\apred_i$ representing the overloading of the back server~$i$.
We set $i = 1$, $j = 2$, and statistically verify the specification with~\cref{alg:nested2} with $t \in \{0.1, 5.0\}$, $\delta \in \{0.1, 0.5\}$ and $\varepsilon \in \{0.1, 0.5\}$, for the significance level $\alpha = 0.95$.

{From the first row of \cref{tb:qn,tb:qn2}, for both the small and large queueing network, with probability $1 - \varepsilon = 0.9$ for the back server $1$, we have that $\abs{p_1 - p_2} < 0.1$, where $p_1$ and $p_2$ are the probabilities that  back server $2$ overloads $t = 5$ earlier/later than back server $1$, respectively.
This shows that for $90\%$ cases, the shortest-queue-delivery policy is roughly fair for back-end server $1$, in the sense~that its overload time is not significantly sooner or later (for $t = 5$)~than back-end server $2$.
On the other hand, from the third row of \cref{tb:qn,tb:qn2}, with probability $1 - \varepsilon = 0.9$ for back-end server $1$, we have $\abs{p_1' - p_2'} > 0.5$, where $p_1'$ and $p_2'$ are the probabilities that  back server $2$ overloads $t = 0.1$ earlier/later than  back server $1$, respectively.
Thus, for $90\%$ cases, the shortest-queue-delivery policy is not exactly fair for back server $1$, in the sense that its overload time can be moderately sooner or later (for $t = 0.1$) than for back-end server $2$.}

\begin{table}[!t]
\centering
\setlength\tabcolsep{2mm}
\begin{tabular}{rrrrrrrr}
$t$ & $\delta$ & $\varepsilon$ & $\alpha$ & Acc. & Sam. & Time (s) & Ans. \\
\toprule
 0.1 &  0.1 &  0.1 &  0.05 &  1.00 &  1.5e+02 & \yw{8.0e+01} & \False \\
 0.1 &  0.5 &  0.5 &  0.05 &  1.00 &  1.4e+02 & \yw{8.2e+01} & \False \\
 5.0 &  0.1 &  0.1 &  0.05 &  1.00 &  7.3e+02 & \yw{3.9e+02} & \True \\
 5.0 &  0.5 &  0.5 &  0.05 &  1.00 &  3.1e+01 & \yw{1.9e+01} & \True \\
\bottomrule
\end{tabular}
\caption{Accuracy (Acc.), average number of samples (Sam.), \yw{average execution time (Time)}, and SMC results (Ans.) for fairness~\eqref{eq:spec_qn} with parameters $t$, $\delta$ and $\varepsilon$ of a Small Queueing Network (with 1 front-end server and 2 back-end servers) under significance level $\alpha$.\label{tb:qn}}\end{table}

\begin{table}[!t]
\centering
\setlength\tabcolsep{2mm}
\begin{tabular}{rrrrrrrr}
$t$ & $\delta$ & $\varepsilon$ & $\alpha$ & Acc. & Sam. & Time (s) & Ans. \\
\toprule
 0.1 &  0.1 &  0.1 &  0.05 &  1.00 &  2.1e+02 & \yw{3.2e+02} & \False \\
 0.1 &  0.5 &  0.5 &  0.05 &  1.00 &  3.3e+02 & \yw{4.9e+02} & \False \\
 5.0 &  0.1 &  0.1 &  0.05 &  1.00 &  6.8e+02 & \yw{1.1e+03} & \True \\
 5.0 &  0.5 &  0.5 &  0.05 &  1.00 &  4.2e+01 & \yw{6.6e+01} & \True \\
\bottomrule
\end{tabular}
\caption{Accuracy (Acc.), average number of samples (Sam.), \yw{average execution time (Time)}, and SMC results (Ans.) for fairness of a Large Queueing Network (with 25 front-end servers and 20 back-end servers) with parameters defined as in \cref{tb:qn}.\label{tb:qn2}}
\end{table}

\section{Conclusion} \label{sec:conc}

In this work, we have studied statistical verification of hyperproperties for cyber-physical systems (CPS).
We have first defined a general model of probabilistic uncertain systems (PUS) that
unify commonly studied modeling formalisms such as continuous-time Markov
chains and hybrid I/O automata with probabilistic dynamical parameters.
We have then introduced hyperproperties, such as fairness and sensitivity, that
involve relationships between multiple paths simultaneously in continuous time.
To formally specify such hyperproperties, we have introduced Hyper Probabilistic
Signal Temporal Logic (\hpstl), which is a hyper and probabilistic version of
the conventional the signal temporal logic (STL).
To simplify our  presentation of the \hpstl syntax and semantics, while considering SMC of specific hyperproperties of complex CPS, in the proposed \hpstl, we have disallowed nested existential and universal quantifies on states to avoid exhaustive iteration on the possibly infinite state space. Still, this logic can be augmented (as done in detail in~\cite{wang2019statistical}) by allowing nested existential and universal quantifications over multiple states.

To verify \hpstl specifications on the PUS, we have developed statistical model 
checking (SMC) algorithms with three new features:
(1)~the significance level of the \hpstl specifications are computed directly 
using the Clopper-Pearson significance level;
(2)~statistically verifying \hpstl specifications on the joint probabilistic 
distribution of multiple paths,
and (3)~\hpstl specifications with nested probabilistic operators quantifying 
different paths are allowed.
Finally, we have evaluated the introduced SMC algorithms on different CPS 
benchmarks with varying levels of complexity.

The work in this paper opens many new avenues for further research. Our work 
has direct application in software doping~\cite{dbbfh17} and in particular, 
identifying whether two black-box systems are statistically identical. 
This problem is in particular challenging in the CPS domain, since the physical environment may make the behavior of 
cyber components more unpredictable. Another application area of our work is in 
conformance-based testing~\cite{amf14}. We also plan to expand our work to 
analyzing information-flow security and in particular differential~privacy.

\begin{acks}
This work is sponsored in part by the ONR under agreements N00014-17-1-2504, AFOSR under award number FA9550-19-1-0169, as well as the NSF CNS-1652544 and NSF SaTC-1813388 grants.
\end{acks}

\bibliographystyle{ACM-Reference-Format}

\end{document}